\pdfoutput=1
\documentclass[10pt, letterpaper, conference]{IEEEtran}

\ifCLASSINFOpdf
	\usepackage[pdftex]{graphicx}
	\graphicspath{{./Figures/}}
\else
  \usepackage[dvips]{graphicx}
  \graphicspath{{./Figures/}}
\fi
	
\usepackage[cmex10]{amsmath}
\usepackage[caption = false, font = footnotesize]{subfig}
\usepackage{amsthm}
\usepackage{amsfonts,url}

\newtheorem{theorem}{Theorem}
\newtheorem{lemma}{Lemma}

\newcommand*{\Rom}[1]{\uppercase\expandafter{\romannumeral #1\relax}}  
\usepackage{textcomp}
\usepackage{gensymb} 
\usepackage{bbm}
\usepackage{authblk}
\usepackage{stackengine}
\usepackage{changes} 
\setaddedmarkup{{\color{orange} #1}} 
\usepackage{multirow} 
\usepackage{bm}

\usepackage{multirow}

\DeclareMathOperator{\Reals}{\mathbb{R}}

\newcommand{\brac}[1]{\left(#1\right)}
\newcommand{\sbrac}[1]{\left[#1\right]}
\newcommand{\cbrac}[1]{\left\{#1\right\}}

\newcommand{\bvec}[1]{\boldsymbol{#1}}
\newcommand{\bvecgreek}[1]{\bm{#1}}
\newcommand{\expect}[1]{\mathbb{E}\sbrac{#1}}

\newcommand{\abs}[1]{\left| #1 \right|}

\newcommand{\rbartime}[3]{\overline{r}_{#1}^{\epsilon}\brac{#2,#3} }

\DeclareMathOperator*{\E}{\mathrm{E}}

\theoremstyle{assumption}
\newtheorem{assumption}{Assumption}
\theoremstyle{definition}
\newtheorem{definition}{Definition}
\theoremstyle{statement}
\theoremstyle{corollary}
\newtheorem{corollary}{Corollary}

\begin{document}

\title{Joint Scheduling of URLLC and eMBB Traffic in 5G Wireless Networks }


 \author[*]{Arjun Anand}
\author[*]{Gustavo de Veciana}
\author[*]{Sanjay Shakkottai}
\affil[*]{Department of Electrical and Computer Engineering, The University of Texas at Austin}

\maketitle


\thispagestyle{plain}
\pagestyle{plain}

\begin{abstract}
Emerging 5G systems will need to efficiently support both enhanced mobile broadband
traffic (eMBB) and ultra-low-latency communications (URLLC) traffic.  In these
systems, time is divided into slots which are further sub-divided into
minislots. From a scheduling perspective, eMBB resource allocations
occur at slot boundaries, whereas to reduce latency URLLC traffic is
pre-emptively overlapped at the minislot timescale, resulting in
selective superposition/puncturing of eMBB allocations. This approach
enables minimal URLLC latency at a potential rate loss to eMBB
traffic.

We study joint eMBB and URLLC schedulers for such systems, with the
dual objectives of maximizing utility for eMBB traffic while immediately
satisfying   URLLC demands. For a linear rate loss model
(loss to eMBB is linear in the amount of URLLC superposition/puncturing), we
derive an optimal joint scheduler. Somewhat counter-intuitively, our
results show that our dual objectives can be met by an iterative
gradient scheduler for eMBB traffic that anticipates the expected loss
from URLLC traffic, along with an URLLC demand scheduler that is
oblivious to eMBB channel states, utility functions and allocation
decisions of the eMBB scheduler.  Next we consider a more general
class of (convex/threshold) loss models and study optimal online joint
eMBB/URLLC schedulers within the broad class of channel state
dependent but minislot-homogeneous policies. A key observation is that unlike the linear rate loss model, for the convex and threshold rate loss models, optimal eMBB and URLLC scheduling decisions do not de-couple and joint optimization is necessary to satisfy the dual objectives.   We validate the
characteristics and benefits of our schedulers via simulation.

\end{abstract}



\section{Introduction}
\label{sec:intro}

An important requirement for 5G wireless systems is its ability to
efficiently support both broadband and ultra reliable low-latency 
communications. On one hand enhanced Mobile
Broadband (eMBB) might require  gigabit per second data rates {
  (based on a bandwidth of several 100 MHz)} and a  moderate latency (a few
milliseconds). On the other hand, Ultra Reliable Low Latency
Communication (URLLC) traffic requires extremely low delays (0.25-0.3
msec/packet) with very high reliability
(99.999\%)~\cite{3gpp_ran1_87}. To satisfy these heterogenous
requirements, the 3GPP standards body has proposed an innovative
\emph{superposition/puncturing} framework for multiplexing URLLC and
eMBB traffic in 5G cellular systems\footnote{An earlier version of this work appears in the Proceedings of IEEE Infocom 2018, Honolulu, HI,~\cite{AndeVSh18}.}.

The proposed scheduling framework has the following structure
\cite{3gpp_ran1_87}.  As with current cellular systems, time is
divided into slots, with a  proposed one millisecond (msec) slot
duration. Within each slot, eMBB traffic can share the bandwidth over
the time-frequency plane (see
Figure~\ref{fig:time_frequency_plane}). The sharing mechanism can be
opportunistic (based on the channel states of various users); however,
the eMBB shares are decided by the beginning, and fixed for the
duration of a slot\footnote{The sharing granularity among various eMBB
  users is at the level of Resource Blocks (RB), which are small
  time-frequency rectangles within a slot. In LTE today, these are (1
  msec $\times$ 180 KHz), and could be smaller for 5G systems.}.
Further the new framework also allows aggregation of eMBB slots where transmissions to an eMBB user over consecutive slots are coded together  to achieve better coding gains resulting from long codewords while reducing  overheads due to control signals. This results in better spectral efficiency as compared to the OFDMA frame structure of LTE~\cite{pbfms16}. 

\begin{figure}
\centering
\includegraphics[width=2.5in]{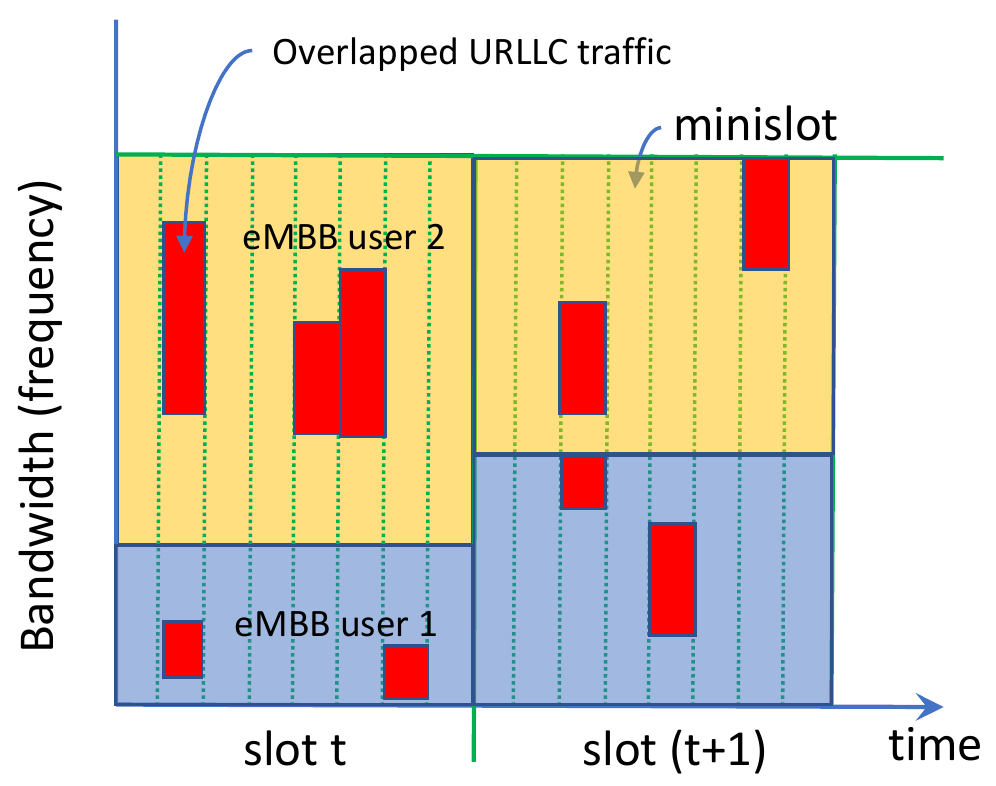}
\caption{Illustration of superposition/puncturing approach for
  multiplexing eMBB and URLLC: Time is divided into slots, and further
  subdivided into minislots. eMBB traffic is scheduled at the
  beginning of slots (sharing frequency across two eMBB users),
  whereas URLLC traffic can be dynamically overlapped
  (superpose/puncture) at any minislot.}
\label{fig:time_frequency_plane}
\end{figure} 

URLLC downlink packets may arrive during an ongoing eMBB transmission;
if tight latency constraints are to be satisfied, they cannot be
queued until the next slot.  Instead each eMBB slot is divided into
minislots, each of which has a 0.125 msec duration\footnote{In 3GPP,
  the formal term for a `slot' is eMBB TTI, and a `minislot' is a
  URLLC TTI, where TTI expands to Transmit Time Interval.}. Thus upon
arrival URLLC packets can be immediately scheduled in the next minislot
{\em on top of the ongoing eMBB transmissions.} If the Base Station (BS)
chooses non-zero transmission powers for both eMBB and overlapping
URLLC traffic, then this is referred to as {\em superposition.} If
eMBB transmissions are allocated zero power when URLLC traffic is
overlapped, then it is referred to as {\em puncturing} of eMBB
transmissions. To achieve high reliability URLLC transmissions are by design 
protected through coding and HARQ if necessary.
At the end of an eMBB slot, the BS can signal  eMBB users the
locations, if any, of URLLC superposition/puncturing.  eMBB users
can then use this information to decode transmissions, with some
possible loss of rate depending on the amount of URLLC overlap.
See~\cite{3gpp_ran1_87, 3gpp_ran1_88} for additional details.

A key problem in this setting is  the {\em joint scheduling of
  eMBB and URLLC traffic over two time-scales.}  At the slot boundary,
resources are allocated to eMBB users (with possible aggregation of slots) based on their channel states 
and utilities, in effect, allocating long term rates 
to optimize high-level goals (e.g. utility optimization). 
Meanwhile, at each minislot boundary, the (stochastic) URLLC
demands are placed  onto previously scheduled and ongoing
eMBB transmissions. Decisions on the placement of such overlaps 
across scheduled eMBB user(s) will impact the rates they will see
on that slot. Thus we have a coupled problem of jointly optimizing
the scheduling of eMBB users on slots with the placement of URLLC demands 
across minislots.

\subsection{Main Contributions}
\label{sec:main-contribs}

This paper is, to our knowledge, the first to formalize and solve the
joint eMBB/URLLC scheduling problem described above. We consider
various models for the eMBB rate loss associated with URLLC
superposition/puncturing, for which we characterize the associated
feasible throughput regions and propose online joint scheduling algorithms as
detailed below.


\noindent \textbf{Linear Model:} When the rate loss to eMBB is
directly proportional to the fraction of superposed/punctured
minislots, we show that the joint optimal scheduler has a nice
decomposition.
Despite having non-linear utility functions and time-varying channel
states, the stochastic URLLC traffic can be {\em uniform-randomly
  placed} in each minislot, while the eMBB scheduler can be scheduled
via a greedy iterative gradient algorithm that only accounts for the
{\em expected} rate loss due to the URLLC traffic.

\noindent \textbf{Convex Model:} For more general settings where the
rate loss can be modeled by  a convex function, the solution does
not have the decomposition property as in the linear model and hence,
the finding the  optimal solution is challenging. Therefore,  we
restrict to a simpler class of joint scheduling policies called as
\emph{minislot-homogeneous} joint scheduling policies where  the URLLC
placement policy does not change across the minislots in an eMBB
slot. In this setting, we characterize the capacity region and derive
concavity conditions under which we can 
derive the effective rate seen by eMBB users (post-puncturing by URLLC
traffic). We then develop a stochastic approximation algorithm which jointly
schedules eMBB and URLLC traffic, and show that it asymptotically
maximizes the utility for eMBB users while satisfying URLLC demands. We
also show that for convex  functions which are
\emph{homogeneous}, minislot-homogeneous joint scheduling
policies are optimal within  the larger class of \emph{causal} and
\emph{non-anticipative} joint scheduling policies.  
Further for the convex loss model, we show that it is better to
schedule eMBB users to share bandwidth (i.e. slice across frequency,
see also Fig.~\ref{fig:config_2}), and let  each user occupy the entire
slot duration  to mitigate rate loss due to URLLC puncturing.


\noindent \textbf{Threshold Model:} Finally we consider a  loss
model, where eMBB traffic is unaffected by puncturing until a
threshold is reached; beyond this threshold it suffers complete throughput loss
(a 0-1 rate loss model). We consider two broad classes of minislot
homogeneous policies, where the URLLC traffic is placed in minislots in
proportion to the eMBB resource allocations (Rate
Proportional (RP)) or eMBB loss thresholds (Threshold Proportional (TP)). We motivate these policies (e.g. TP minimizes the probability of any eMBB loss in an eMBB slot) and derive the associated throughput regions.  Finally, we 
utilize the additional structure underlying the RP and TP Placement
policies along with the shape of the threshold loss function to
derive fast gradient algorithms that converge and provably maximize
utility. 


\subsection{Related Work}
\label{sec:related}

Resource allocation, utility maximization and opportunistic scheduling
for downlink wireless systems have  been intensely  studied in the last two
decades, and have had a major impact on cellular standards. We refer to
\cite{sriyin14,GeoNeeTas06} for a survey of the key results. In this
paper, we focus on joint scheduling of URLLC and eMBB traffic. From an
application point of view, there have been several studies arguing for
the need to support URLLC services (e.g. for industrial automation)
\cite{hwwtahaa16,ywjbas15,gla17}.

With demand of both broadband and low-latency services growing, there
has been rapid developments in the 5G standardization efforts in 3GPP.
Of key relevance to this paper, the 3GPP RAN WG1 has focused on
standardizing slot structure for eMBB and URLLC, and have been
evaluating signaling and control channels to support superposition and
puncturing in recent meetings \cite{3gpp_ran1_87,3gpp_ran1_88}. We
specifically refer the reader to Sections~8.1.1.3.4 -- 8.1.1.3.6 in
\cite{3gpp_ran1_88} for current proposals.

Beyond standards, recent work has focused on system level design for
such systems (overheads, packet sizes, control channel structure,
etc.) \cite{pbfms16,ljcjs17,dkp16}. Of particular note, \cite{ljcjs17}
argues (based on system level simulation and queuing models) that
statically partitioning bandwidth between eMBB and URLLC is very
inefficient. There have also been several studies focusing on
physical layer aspects of URLLC (coding and modulation, fading, link
budget) \cite{dkopy16,sltum16}.

 Efficient  sharing of radio resources between eMBB and URLLC traffic has been discussed in literature, see~\cite{ylpy18, ptsd18, ksp18}. In~\cite{ylpy18}, the authors have considered joint optimization of  resource allocation for eMBB and URLLC traffic. However, they do not use puncturing/superposition mechanisms to share resources. Some works (\cite{ptsd18, ksp18}) use information theoretic results to obtain expressions for the average eMBB rates under URLLC puncturing for various decoding schemes for uplink eMBB traffic punctured/superposed by URLLC users. However, they do not consider the design of joint scheduling algorithms for eMBB and URLLC traffic.   To the best of our knowledge, our paper is the first to explore the resource allocation issues for joint scheduling of URLLC and eMBB traffic using puncturing/superposition based mechanisms.

\section{System Model}
\label{sec:sys-model}



{\em \bf  Traffic model:}
We consider a wireless system supporting a fixed set $ \mathcal{U} $ 
of backlogged eMBB users and a
stationary process of URLLC  demands. eMBB scheduling decisions are made across slots while URLLC demands arrive  and are immediately scheduled in the next minislot.  In this section we shall consider the case where eMBB all users receive resources for slots without using slot aggregation even though more flexible resource allocations which can possibly include slot aggregation and splitting are proposed in 5G standards~\cite{pbfms16}. We shall justify this choice in Sec.~\ref{sec:eMBB_slot_aggregation}. 
Each eMBB slot has an associated set of 
minislots where the set ${\cal M} = \{ 1,\dots |{\cal M}| \}$ denotes
their indices.   
URLLC demands across minislots are modeled as an independent and identically distributed (i.i.d.) 
random process. We let the random variables 
$(D(m), m \in {\cal M})$ denote the URLLC demands per minislot for a typical eMBB slot and 
let $D$ be a random variable whose distribution is that of the aggregate URLLC demand per eMBB slot, i.e.,
$D \sim \sum_{m \in {\cal M}} D(m)$ with, cumulative distribution function $F_D(\cdot)$ and  mean $E[D] = \rho.$ 
We assume demands have been normalized so the maximum URLLC demand per minislot is $f$ and the maximum aggregate demands per eMBB slot is 
$f \times |{\cal M}| = 1$ i.e., all the frequency-time resources are occupied.
URLLC demands per minislot exceeding the system capacity are blocked by URLLC scheduler thus $D \leq 1$ almost surely. 
The system is engineered so that blocked URLLC traffic on a minislot 
 is a rare event, i.e., satisfies the desired reliability on such traffic. 

{\em \bf Wireless channel variations:}
The wireless system experiences channel variations each eMBB slot which are modeled as an i.i.d. 
random process over a set of channel states ${\cal S} = \{1,\ldots, |{\cal S}|\}.$  
Let $S$ be a random variable modeling the distribution over the states in a typical eMBB slot 
with probability mass function $p_S(s) = P(S=s)$ for $s \in {\cal S}.$  
For each channel state $s$ eMBB user $u$ has a known peak rate $\hat{r}_{u}^s.$ 
The wireless system can choose what proportions of the frequency-time resources to allocate to
each eMBB user on each minislot for each channel state. This is modeled by a 
 matrix ${\bm \phi} \in \Sigma $ where
\begin{multline}
 \Sigma  :=  \left \{ \bvecgreek{\phi} \in \Reals_+^{|{\cal U}|\times|{\cal M}|\times|{\cal S}|}
 ~|~ \right. \\ \left.   \sum_{u \in {\cal U}}  \phi_{u,m}^{s} =f,  \forall m \in {\cal M}, s \in {\cal S} \right\} 
\end{multline}
and where the element $\phi_{u,m}^s$ represents the fraction of resources allocated to user $u$
in mini slot $m$ in channel state $s$. We also let $\phi_{u}^{s} = \sum_{m\in {\cal M}} \phi_{u,m}^{s}$, i.e., the total resources
allocated to user $u$ in an eMBB slot in channel state $s$.
Now assuming no superposition/puncturing if the system is in channel state $s$ and the eMBB scheduler chooses 
an allocation ${\bm \phi}$ the rate $r_u$ allocated to user $u$ would be given by  
$ r_u = \phi_{u}^s \hat{r}_{u}^{s}. $
The scheduler is assumed to know the channel state and can
thus opportunistically exploit such variations  in allocating resources 
to eMBB users. Note that for simplicity, we adopt a flat-fading model,
namely, the rate achieved by an user is directly proportional to 
the fraction of bandwidth allocated to it (the scaling factor is the
peak rate of the user for the current channel state).



{\em \bf Class of joint eMBB/URLLC schedulers:}  We consider a class
of stationary joint eMBB/URLLC schedulers denoted by $\Pi$ satisfying
the following properties.  A scheduling policy combines a possibly
state dependent eMBB {\em resource allocation} matrix $\bvec{ \phi}$ per slot
with a URLLC {\em demand placement} strategy across minislots.  The
placement strategy may impact the eMBB users' rates since it affects
the URLLC superposition/puncturing loads they will experience.  As
mentioned earlier in discussing the traffic model, in order to meet
low latency requirements URLLC traffic demands are scheduled
immediately upon arrival or blocked.  The scheduler is assumed to be
{\em causal} so it only knows the current (and past) channel states
and peak rates
$ \hat{r}_{u}^{s}$ for all  $u \in {\cal U}$ and  $s \in {\cal S}$ but does
not know the realization of future channels or URLLC traffic demands.
In making superposition/puncturing decisions across minislots, the
scheduler can use knowledge of the previous placement decisions that
were made.  In addition the scheduler is assumed to know (or able
measure over time) the channel state distribution across eMBB slots
and URLLC demand distributions per minislot i.e., that of $D(m)$, and
per eMBB slot, i.e., $D$, and thus in particular knows $\rho = E[D]$.

In summary a joint scheduling policy $\pi \in \Pi$ is thus characterized by the following:
\begin{itemize}
\item an eMBB resource allocation ${\bm \phi}^{\pi} \in \Sigma$ where  
${\phi}^{\pi,s}_{u,m}$ denotes the fraction of frequency-time slot resources allocated to eMBB 
user $u$ on minislot $m$ when the system is in state $s$.
\item the distributions of URLLC loads across eMBB resources induced by its URLLC placement
strategy, denoted by random variables
${\bf L}^{\pi} = (L^{\pi,s}_{u,m} | u \in {\cal U}, m \in {\cal M}, s \in {\cal S})$ where
$L^{\pi,s}_{u,m}$ denotes the URLLC load superposed/puncturing the resource
allocation of user $u$ on minislot $m$ when the channel is in state $s$.  
\end{itemize}

The distributions of $L^{\pi,s}_{u,m}$ and their associated means $\overline{l}_{u,m}^{\pi,s}$
depend on the joint scheduling policy $\pi$, but for all states, users and minislots
satisfy
$$
L^{\pi,s}_{u,m} \leq {\phi}^{\pi,s}_{u,m}~~~  \mbox{almost surely}.
$$
In the sequel we let $L^{\pi,s}_{u} = \sum_{m \in {\cal M}} L^{\pi,s}_{u,m}$,
i.e., the aggregate URLLC traffic superposed/puncturing user $u$ in channel state $s$,
and denote its mean by $\overline{l}_{u}^{\pi,s}$ and note that
$$
L^{\pi,s}_{u} \leq {\phi}^{\pi,s}_{u} \quad  \mbox{almost surely}.
$$
We also let $L^{\pi,s} := \sum_{u \in {\cal U}} L^{\pi,s}_{u}$
denote the aggregate induced load and note that 
any policy $\pi$ and for any state $s$ we have that
$$ 
\rho = \E[D] = \E[L^{\pi,s }] = \E[\sum_{u \in {\cal U}}  L_{u}^{\pi,s}] = \sum_{u \in {\cal U}}  \overline{l}_{u}^{\pi,s}.
$$


{\em \bf Modeling superposition/puncturing and eMBB capacity regions:}
Under a joint scheduling policy $\pi$ 
we model the rate achieved by an eMBB user $u$ in channel state $s$ 
by a random variable  
\begin{eqnarray}
R^{\pi,s}_{u} &=& f_{u}^{s} ( \phi^{\pi,s}_{u}, L_{u}^{\pi,s}), \label{eqn:tput-fn}
\end{eqnarray}
where the {\em rate allocation function} $f_{u}^{s}(\cdot ,\cdot )$ models the impact of URLLC
superposition/puncturing -- one would expect it to be  
increasing in the first argument (the allocated resources)
and decreasing  in the second argument (the amount superposition/puncturing by URLLC traffic). 
 Under our system model we have that 
$$
R^{\pi,s}_{u} \leq  f_{u}^{s} ( \phi^{\pi,s}_{u},0) = \phi^{\pi,s}_{u} \hat{r}_{u}^{s} ~\mbox{almost surely},
$$
with equality if there is no superposition/puncturing, i.e., when $l_{u}^{s}=0.$ 
Let  $\overline{r}^{\pi,s}_{u} = E[ R^{\pi,s}_{u}]$ denote the mean rates achieved by user $u$ in state $s$
under the URLLC superposition/puncturing distribution induced by scheduling policy $\pi$.

\noindent {\bf Models for Throughput Loss:} In the sequel we shall
consider specific forms of superposition/puncturing loss models: {\em (i)}
linear, {\em (ii)} convex, and {\em (iii)} threshold models.

We rewrite the rate allocation function in (\ref{eqn:tput-fn}) as the
difference between the peak throughput and the loss due to URLLC
traffic, and consider functions that can be decomposed as:
$$
f_{u}^{s}(\phi_{u}^{s},l_{u}^{s}) =  \hat{r}_{u}^{s} \phi_{u}^{s}
\left(1-  h^s_u \left( \frac{L_{u}^{\pi,s}}{\phi_{u}^{s}}\right)\right), 
$$
where $h^s_u :[0,1] \rightarrow [0,1]$ is the {\em rate loss function}
and captures the relative rate loss due to URLLC overlap on eMBB allocations. 
The puncturing models we study now map directly to structural
assumptions on the rate loss function $h^s_u(\cdot);$ namely it is a
non-decreasing function, and is one of {\em linear, convex, or
  threshold} as shown in 
Figure~\ref{fig:lossfunctions}.

\begin{figure}
\centering
\includegraphics[height=1.5in,width=3.5in]{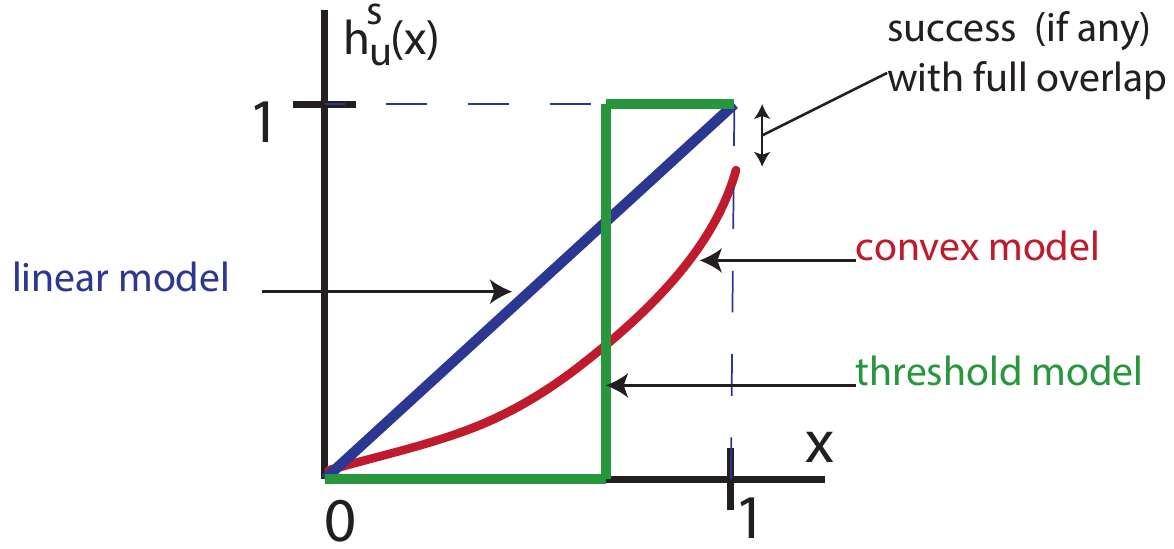}
\caption{The illustration exhibits the rate loss function for the
  various models considered in this paper, linear, convex and
  threshold.}
\label{fig:lossfunctions}
\end{figure}


\noindent {\bf Linear Model:} Under the linear model, the expected rate
for user $u$ in channel state $s$ for policy $\pi$ is given by
$$
r^{\pi,s}_{u} = \E [f_{u}^{s} ( \phi^{\pi,s}_{u}, L_{u}^{\pi,s})]=
\hat{r}_{u}^{s}( \phi^{\pi,s}_{u} - \overline{l}_{u}^{\pi,s}), 
$$
i.e., $h^s_u(x) = x,$ and the resulting rate to eMBB users is
a linear function of both the allocated resources and mean
induced URLLC loads.  This model is motivated by basic results for the
channel capacity of AWGN channel with erasures, see~\cite{Julian_2002}
for more details. Our system in a given network state can be
approximated as an AWGN channel with erasures, when the slot sizes are
long enough so that the physical layer error control coding of eMBB
users use long code-words. Further, there is a dedicated control
channel through which the scheduler can signal to the eMBB receiver
indicating the positions of URLLC overlap. Indeed such a control
channel has been proposed in the 3GPP
standards~\cite{3gpp_ran1_87}. Note that under this model the rate
achieved by a given user depends on the aggregate
superposition/puncturing it experiences, i.e., does not depend on
which minislots and frequency bands it occurs.  We discuss scheduling
policies for  linear loss models in Section~\ref{sec:erasure_channel}.

\noindent {\bf Convex Model:} In the convex model, the rate
loss function $h^s_u(\cdot)$ is convex (see
Figure~\ref{fig:lossfunctions}), and the resulting rate for eMBB user
$u$ in channel state $s$ under policy $\pi$ is given by
$$
r^{\pi,s}_{u} = \E [f_{u}^{s} ( \phi^{\pi,s}_{u}, L_{u}^{\pi,s})]= 
\hat{r}_{u}^{s} \phi^{\pi,s}_{u} \left(1- E\left[ h^s_u \left( \frac{L^{\pi,s}_{u}}{\phi_{u}^{\pi,s}}\right)\right]\right).
$$
This covers a broad class of models, and is discussed in
Section~\ref{sec:convex_model}. 

\noindent {\bf Threshold Model:} Finally the threshold model is designed to capture a simplified 
packet transmission and decoding process in an eMBB receiver. The data
is either received perfectly or it is lost depending on the amount of
superposition/puncturing.   
With slight abuse of notation we shall let $h^s_u$ also depend on both the relative URLLC load
and the eMBB user allocation, i.e., $h^s_u (x) = {\bf 1} (x \geq t_u^s(\phi^s_u))$ where 
the threshold in turn is an increasing function $t_{u}^{s}(\cdot)$ satisfying  $x \geq t_{u}^{s}(x) \geq 0.$
Such thresholds might reflect various engineering choices where codes are adapted when users are 
allocated more resources, so as to be more robust to interference/URLLC superposition/puncturing.
The resulting rate for eMBB user $u$ in channel state $s$ and policy $\pi$ is then given by 
$$
r^{\pi,s}_{u} = \hat{r}_{u}^{s} \phi^{\pi,s}_u P(L_{u}^{\pi,s} \leq  \phi_u^{\pi,s} t_{u}^{s}(\phi_{u}^{\pi,s})).
$$ 
While such a sharp falloff is somewhat extreme, it is nevertheless
useful for modeling short codes that are designed to tolerate a limited
amount of interference.  In practice one might expect a smoother fall
off, perhaps more akin to the convex model, e.g., when hybrid ARQ
(HARQ) is used.  We discuss polices under the threshold based model in
Section~\ref{sec:threshold_model}.



\noindent {\bf Capacity set for eMBB traffic:}
We define the capacity set ${\cal C} \subset \Reals_+^{|{\cal U}|}$
for eMBB traffic as the set of long term rates achievable under policies in $\Pi.$ 
Let ${\bf c}^\pi = ( c_u^\pi | u \in {\cal U} )$ where
$$
c_u^\pi = \sum_{s \in {\cal S}} r^{\pi,s}_{u} p_S(s).
$$
Then the capacity is given by 
$$
{\cal C} = 
\{ \mathbf{c} \in \Reals_+^{|{\cal U}|} ~|~ \exists~ \pi \in \Pi  ~\mbox{such that}~ \mathbf{c} \leq \mathbf{c}^{\pi} \}.
$$
Note that this capacity region depends on the scheduling policies under consideration
as well as the distributions of the channel states and URLLC demands. 

\noindent {\bf Scheduling objective: URLLC priority and eMBB utility
  maximization:} As mentioned earlier, URLLC traffic is immediately
scheduled upon arrival, in the next minislot, i.e, no queuing is
allowed. Thus if demands exceed the system capacity on a given
minislot  traffic would be lost. However, we assume that the system has been engineered so that such
URLLC overloads are extremely rare, and thus URLLC traffic can meet
extremely low latency requirements with high reliability\footnote{Note that since we allow URLLC traffic in the entire system bandwidth,  such overload events are very rare. }.  For eMBB
traffic we adopt a utility maximization framework wherein each eMBB
user $u$ has an associated utility function $U_u(\cdot)$ which is a
strictly concave, continuous and differentiable of the average rate
$c_u^\pi$ experienced by the user.  Our aim is to
characterize optimal rate allocations associated with the utility
maximization problem:
\begin{equation}
\max_{\mathbf{c}} \{  \sum_{u \in \mathcal{U}} U_u\brac{c_u}~|~ \mathbf{c} \in {\cal C}  \},
\end{equation} 
and determine a scheduling policy $\pi$ that will realize such allocations.

\section{Linear Model for Superposition/Puncturing}
\label{sec:erasure_channel}

In any state $ s $, the
optimal joint eMBB/URLLC scheduler may either 1) protect the user with
the lower channel rate by placing less URLLC traffic into its
frequency resources to ensure fairness or 2) opportunistically place
URLLC traffic so that the user with a better channel gets a higher
rate to improve the overall system throughput. The solution for any
state is complex function of network states and their distribution and
user utility functions and in general, eMBB scheduling and URLLC
puncturing may be dependent.
In this section, we show a surprising result -- despite having
non-linear utility functions, if the loss functions are linear and the 
eMBB scheduler is intelligent (i.e., takes into the degradation of
rates due to puncturing), then the URLLC scheduler can be {\em
  oblivious to the channel states, utility functions and the actual
  rate allocations of the eMBB scheduler.}


\subsection{Characterization of capacity region}

Let us consider the capacity region for a wireless system 
based on linear superposition/puncturing model
under a restricted class of policies $\Pi^{LR}$ that combine
feasible eMBB allocations ${\bm \phi} \in \Sigma$ 
with random placement of URLLC demands uniformly over the bandwidth across minislots. Note that the notation $ LR $ stands for linear loss model (L) with random (R) placement of URLLC traffic. 
For any $\pi \in \Pi^{LR}$ with eMBB allocation 
$\bm{\phi}^\pi$ the mean induced loads under such
randomization for each state $s \in {\cal S}$ and minislot $m \in {\cal M}$ 
will satisfy $\overline{l}^{\pi,s}_{u,m} = \rho \phi^{\pi,s}_{u,m}.$
Indeed randomization clearly 
leads to an induced loads that are proportional to the eMBB allocations
on a per mini-slot basis, but also per eMBB slot, i.e.,
$\overline{l}^{\pi,s}_{u} = \rho \phi^{\pi,s}_{u}.$ 
Thus for our linear loss model we have that
$$
r^{\pi,s}_u=   \hat{r}_{u}^{s} (\phi_{u}^{\pi,s}-\overline{l}_{u}^{\pi,s}) 
=   \hat{r}_{u}^{s} \phi_{u}^{\pi,s} (1 - \rho) .
$$
Hence the overall user rates achieved
under such a policy are given by
${\bf c}^{\pi} = ( c_u^{\pi} | u \in {\cal U} )$ 
where
$$
c_u^{\pi}  = \sum_{s\in {\cal S}}  \hat{r}_{u}^{s} \phi_{u}^{\pi,s}(1- \rho) p_S(s).
$$
The capacity region associated with policies that use URLLC uniformly randomized placement is thus given by 
\begin{eqnarray*}
{\cal C}^{LR}  & = &
\{ \mathbf{c} \in \Reals_+^{|{\cal U}|} 
~|~ \exists \pi \in \Pi^{LR} ~\mbox{s.t.}~ \mathbf{c} \leq {\bf c}^{\pi} \}  \\
& = & \{ \mathbf{c} \in \Reals_+^{|{\cal U}|} 
~|~ \exists {\bm{\phi}} \in \Sigma ~\mbox{s.t.}~ 
 \mathbf{c} \leq {\bf c}^{\bm{\phi}} \},
\end{eqnarray*}
where we have abused notation by using ${\bf c}^{\bm{\phi}}$ to
represent the throughput achieved under policy $\pi$ that uses eMBB
resource allocation ${\bm{\phi}}$ and uniformly randomized URLLC demand
placement. Finally note that for any fixed $\rho \in (0, 1),$
$ {\cal C}^{LR}$ is a closed and bounded convex region. This is
because an affine map of a convex region remains convex; hence
multiplying the constraints on the capacity region defined by
${\bm \phi}$ by a constant $(1-\rho)$ preserves convexity of the
rate region.

\begin{theorem} 
\label{thm:theorem_on_erasure_channel_based_model}
For a wireless system under the linear superposition/puncturing loss model 
we have that ${\cal C} = {\cal C}^{LR}.$ 
\end{theorem}
The proof is deferred to the Appendix~A. In other words the throughput
${\bf c}^\pi \in {\cal C}$ achieved by any feasible policy
$\pi \in \Pi$ can also be achieved by policy $\pi'$, with a possibly
different eMBB resource allocation policy than $\pi$ but utilizing uniform
random placement of URLLC demands across mini-slots.

\subsection{Utility maximizing joint scheduling}
\label{sec:util-linear}

Given the result in Theorem \ref{thm:theorem_on_erasure_channel_based_model}
we now restate the utility maximization problem as optimizing solely over  
joint scheduling policies that use URLLC random placement policies,
as follows:
\begin{eqnarray*}
\label{eq:optimization_prob_linear}
\max_{\bvec{\phi} \in \Sigma} & & \sum_{u \in \mathcal{U}} U_u ( c_u^{\bm \phi}), \\
\mbox{s.t.} & & c_u^{\bm \phi}= \sum_{s\in {\cal S}} \hat{r}_{u}^{s}
                \phi^{s}_{u} (1- \rho)   p_S(s),~~ \forall u \in {\cal
                U} .
\end{eqnarray*}
The above optimization problem has a strictly concave cost function
and convex constraints. Thus, at face-value, it appears that we can
 apply the gradient scheduler introduced in
\cite{Stolyar05}, which is an online algorithm designed to  converge 
to the solution of  similar optimization problem. This observation is approximately
correct, but subject to two modifications.

First, the setting in \cite{Stolyar05} has deterministic rates in each
channel state. However, in our case, in each channel state, the rates
are stochastic due to  puncturing by URLLC traffic (this results in the $(1 - \rho)$ correction). This can be easily
addressed by modifying the setting in \cite{Stolyar05}; the finite state
and i.i.d. nature of puncturing implies that the proofs in
\cite{Stolyar05} hold with minor modifications; we skip the details.

The second issue is somewhat more nuanced. In current wireless systems
(e.g. LTE) and proposals for 5G systems, a slot is partitioned into a
collection of Resource Blocks (RB), where each RB is a time-frequency
rectangle (1 msec $\times$ 180 KHz in LTE). Importantly, these RBs can
be individually allocated to different eMBB users. If we now apply the
gradient scheduler in \cite{Stolyar05} to our setting, the result will
be that all RBs in a slot will be allocated to the same user. While
this is no-doubt asymptotically optimal, it seems intuitive that
sharing RBs across users even within a slot will lead to better
short-term performance. Indeed this intuition has been explored in the
context of iterative MaxWeight algorithms to provide formal
guarantees, see \cite{BodShaYinSri_10,BodShaYinSri_11}. The high level
idea is that even within a slot, RB allocations are done iteratively, where
future RB allocations need to account for prior rate allocations even
within the same slot. This is formalized below, where we 
describe our  proposed joint eMBB-URLLC scheduler.

{\bf The URLLC scheduler:} As explained in the previous section, the
URLLC scheduler places the URLLC traffic uniformly at random in each
minislot.

    

{\bf The eMBB scheduler:}
Let there be $B$ resource blocks available for allocation every eMBB
slot, indexed by $1,2, \ldots, B$. Let $\overline{R}_u(t-1)$ be the random variable denoting the average
rate received by eMBB user up to eMBB slot $t-1$. Let $ \overline{r}_u(t-1) $ be a realization of $\overline{R}_u(t-1)$.  In any eMBB slot $t$ we schedule an
user $u(b)$ in RB $b$ such that  
\begin{equation}
 u(b) \in \text{argmax} \cbrac{\hat{r}_{u}^{s}
   U_u^{'}\brac{\rbartime{u}{b-1}{t}}, \, u=1,2, \ldots, \mathcal{U}},    \label{eqn:opt-search}
 \end{equation}
 where $\rbartime{u}{b-1}{t}$ is an \emph{estimate} of the average rate
 received by eMBB user $u$ till slot $t$ which is iteratively
 updated as follows: 
 \begin{multline}
 \rbartime{u}{b}{t} =
 \begin{cases}
 \overline{r}_u(t-1), & b=0, \\
  \brac{1-\epsilon}\rbartime{u}{b-1}{t} \quad &\\ +  \epsilon
  \brac{\hat{r}_{u}^{s}\frac{1}{B}(1-\rho) \mathbbm{1}\brac{i=u(b)}}, & b \neq 0. \label{eqn:lin-rate-update}
  \end{cases}
 \end{multline}
 In the above equation, $\epsilon$ is a small positive value. At the
 end of eMBB slot $t$, the eMBB scheduler receives feedback from the
 eMBB receivers indicating the actual rates received by the eMBB users
 due to allocations. We denote
 the rate received eMBB user $u$ in slot by the random variable
 $R_u(t)$ and its realization by $ r_u(t) $.  We finally update $\overline{r}_u(t)$ as follows:
\begin{equation}
\label{eq:update_for_rbar_lin}
\overline{r}_u(t)= \brac{1-\epsilon}\overline{r}_u(t-1) + \epsilon r_u(t).
\end{equation} 
This scheduler and update equations  are analogous to the gradient algorithm \cite{Stolyar05}
(see also iterative algorithms in
\cite{BodShaYinSri_10,BodShaYinSri_11}). The optimality proof of this
algorithm follows (with minor modifications) from the analysis in
\cite{Stolyar05}; we skip the details.

\noindent {\bf Remarks:} \textbf{(i)} A natural decomposition of the
joint eMBB+URLLC scheduling is now apparent. On one hand, the eMBB
scheduler maximizes utilities based on the {\em expected} channel
rates stemming from uniform random puncturing of minislots
(accounted for through the $(1-\rho)$ multiplicative factor), and does
so using the iterative gradient scheduler. The URLLC scheduler, on the
other-hand, is completely agnostic to either the channel state or the
actual eMBB allocations and simply punctures minislots based on the
current instantaneous demand.

\textbf{(ii)} The fact that the URLLC traffic placement is completely agnostic
to the channel state and eMBB utilities/allocation is surprising.
Intuitively it seems plausible that one could puncture an eMBB user with a
lower marginal utility with more URLLC traffic, while protecting an
eMBB user with a higher marginal utility and achieve a better sum
utility. Further, it seems reasonable that eMBB users with a worse
channel state (and thus lower rate) could be loaded with additional
URLLC traffic. However,
Theorem.~\ref{thm:theorem_on_erasure_channel_based_model} implies that
there exists an optimal solution that is achieved by channel and
utility oblivious and uniform random URLLC placement, thus providing
a very simple algorithm for URLLC scheduling.

\textbf{(iii)} We remark that the optimality of random puncturing for
linear loss models depends critically on the use of an opportunistic
scheduler for eMBB traffic.To see this, consider a simple system with
two symmetric eMBB users each with two possible channel states. The
associated channel rates are either $\{2, 4\}$ packets/slot with equal
probability, and independent across users and time slots. Suppose that
we use a static (non-opportunistic) scheduler, which equally splits
channel access between the users. It is easy to calculate that the
rate to each user is then 1.5 packets/slot. Next suppose that the
URLLC load is 50\%, and that this traffic {\em randomly punctures}
eMBB users. Then from symmetry, it follows that the rate per eMBB user
is $0.75$ packets/slot. In contrast, suppose that puncturing is
opportunistic, where the user with the currently lower rate is
punctured whenever possible (opportunistic puncturing of the currently
worse eMBB user), a straightforward calculation shows that the rate to
each eMBB user is $0.875$ packets/slot, which is a {\em strict
  improvement over random puncturing.} At a high-level, this follows
because opportunistic eMBB scheduling operates on the Pareto frontier
of two-user capacity region, and consequently there is no residual
opportunistic to be obtained by puncturing. However, with
non-opportunistic scheduling, the system is not pushed to the
boundary; thus, opportunistic puncturing can extract additional
throughput for eMBB users.

\section{Convex Model -- Minislot-Homogenous Policies}
\label{sec:convex_model}

In this section we shall consider joint scheduling for wireless
systems for convex superposition/puncturing  loss models.  This is a
somewhat complex problem, whence we will focus our attention on a
restricted, but still rich, class of scheduling policies which we
refer to as minislot-homogeneous eMBB/URLLC schedulers. We identify a key
concavity requirement in Assumption~\ref{condn:conc-g} (that is
satisfied by convex loss functions) that enables a stochastic
approximation approach for utility maximizing scheduling.


\subsection{Minislot-homogeneous eMBB/URLLC Scheduling policies} 
\label{subsec:minislot_homogeneous}

We shall define minislot-homogeneous eMBB/URLLC schedulers as
follows. First, feasible eMBB allocations ${\bm \phi} \in \Sigma$ will
be restricted such that for any eMBB slot in channel state
$s \in {\cal S}$ allocations are {\em minislot-homogeneous} across
minislots in an eMBB slot, i.e.,
$\phi_{u,1}^{s}=\phi_{u,m}^{s}, \forall m \in {\cal M}$ and its
overall allocation for the slot is given by
$\phi_{u}^{s} = |{\cal M}| \phi_{u,1}^{s}.$ The set of
minislot-homogeneous eMBB allocations is thus given by
\begin{multline} \nonumber \Sigma^H := \left \{ \bvecgreek{\phi} \in \Sigma
    ~|~  u\in {\cal U},~ \phi_{u,m}^{s} =\phi_{u,1}^s
    ~~ \forall m \in {\cal M},   \forall s \in {\cal S}\right\}.
\end{multline}

Second, URLLC demand placements per minislot are done proportionally based on
pre-specified weights, and these weights are assumed to be
time-homogeneous across minislots. In particular such policies are
parametrized by a weight matrix $\bvecgreek{\gamma} \in \Sigma^H$, where the induced
load on user $u$ under channel state $s$ and slot $m$ is given by
$$
L_{u,m}^{s} = \frac{\gamma_{u,m}^s}{\sum_{u' \in {\cal U}}
  \gamma_{u',m}^s} D(m) = \frac{\gamma_{u,1}^s}{f} D(m).
$$
We shall call $ \gamma_{u,1}^s$ the \emph{URLLC placement factor} for eMBB user $ u $ in state $ s $. 
The eMBB and URLLC allocations are coupled together since it
must be the case that for all $u\in {\cal U}$
$L_{u,m}^{s} \leq \phi_{u,m}^{s} = \phi_{u,1}^s$ almost surely, i.e.,
one can not induce more superposition/puncturing load on a user than the
resources it has been allocated on that slot. So the following
condition must be satisfied.  For all
$m \in {\cal M}$ we have that
$$
D(m) \leq \min_{u \in {\cal U}} \frac{\phi_{u,1}^{s}}{\gamma_{u,1}^s} f, \, \, \, \mbox{almost surely}.
$$
Recall that $ f $ denotes the maximum URLLC load per minislot so $D(m) \leq f$ almost surely, thus if
$\frac{\phi_{u,1}^{s}}{\gamma_{u,1}^s} \geq 1$ the above condition will always hold. Yet  if $ \phi_{u,1}^{s}\geq \gamma_{u,1}^s $ for all $ u $, then we have that $ \phi_{u,1}^{s}= \gamma_{u,1}^s $, i.e., there is not flexibility to exploit careful placement of URLLC demands. Hence, we introduce the following assumption: 
\begin{assumption}
\label{sharingfactor-assumption}
We say the system has a $(1-\delta)$ URLLC sharing factor per
minislot if  $D(m) \leq f (1-\delta)$ almost surely for all
$m \in {\cal M}$, where $ \delta \in \brac{0,1} $.  
\end{assumption}
For any $\delta$ the above assumption implies that the {\em peak
  URLLC demand} in an eMBB slot can be at most $ 1-\delta$ which
is lower than maximum possible value of one. Such an assumption is
reasonable as we consider shared resources  which are engineered to
meet the peak URLLC loads while also serving eMBB traffic. Under a
$(1-\delta)$ URLLC sharing factor a minislot-homogeneous eMBB resource
allocation ${\bm \phi}$ and URLLC allocation ${\bm \gamma}$ is will be
feasible if for all $s \in {\cal S}$ we have
$$
(1-\delta) \leq \min_{u \in {\cal U}} \frac{\phi_{u,1}^{s}}{\gamma_{u,1}^s},
$$
which is satisfied as long as
$(1-\delta) \gamma_{u,1}^s \leq \phi_{u,1}^{s}$ for all
$u \in {\cal U}$.  This motivates the following definition:
\begin{definition}
{\em   For a system with a $(1-\delta)$ sharing factor, the feasible minislot-homogeneous
  eMBB/URLLC scheduling policies are parameterized by
  ${\bm \phi},{\bm \gamma} \in \Sigma^H$ such that
  $(1-\delta)\bm{\gamma} \leq {\bm \phi}$. We shall denote the set of
  such policies as follows:
$$
\Pi^{H,\delta} := \{ ({\bm \phi},{\bm \gamma}) ~\mid~ {\bm \phi},{\bm
  \gamma} \in \Sigma^H ~\mbox{and}~ (1-\delta)\bm{\gamma} \leq {\bm
  \phi} \} , 
$$
where $\Pi^{H,\delta}$ is a convex set.}
\end{definition}

\subsection{Characterization of the throughput region} 

In this section we characterize the throughput regions achievable
under time-homogeneous scheduling.

\begin{theorem}
  \label{thm:main-theorem}
  For a system with a $(1-\delta)$ sharing factor and minislot-homogeneous scheduler
  ${\bm \pi} = ({\bm \phi}^{\bm \pi},{\bm \gamma}^{\bm \pi}) \in
  \Pi^{H,\delta}$ the average induced throughput for user 
  $u\in {\cal U}$ in channel state $s \in {\cal S}$ is given by
$$
r^{\bm{\pi},s}_{u} =  \E[ {f}^{s}_u ( \phi^{{\bm \pi},s}_{u}, \gamma^{\bm{\pi},s}_u D )],
$$
and the overall average user throughputs are given by
${\bf c}^{{\bm \pi}}= ( c^{{\bm \pi}}_{u} \mid u \in {\cal U})$ where
$ c^{{\bm \pi}}_{u} = \sum_{s \in {\cal S}} r^{{\bm \pi},s}_{u}
p_S(s).$
\end{theorem}
The proof is included in Appendix~B.
%
%
Based on the above we can define feasible throughput region
constrained to the time-homogeneous policies in $\Pi^{H,\delta}.$ 
First let us define 
\begin{eqnarray*} {\cal C}^{H,\delta} & = & \{ \mathbf{c} \in
  \Reals_+^{|{\cal U}|} ~|~ \exists {\bm \pi} \in \Pi^{H,\delta}
  ~\mbox{s.t.}~ \mathbf{c} \leq {\bf c}^{{\bm \pi}} \}
\end{eqnarray*}
and let $\hat{\cal C}^{H,\delta}$ denote the convex hull of
${\cal C}^{H,\delta}.$ Note that rates in the convex hull
are achievable through policies that do time sharing/randomization
amongst minislot-homogeneous scheduling policies in $\Pi^{H,\delta}.$

\begin{assumption} \label{condn:conc-g}
For all $s\in {\cal S}$ and $u \in {\cal U}$
the functions $g^s_u ( , ) $ given by  
\begin{equation}
\label{eqn:conc-condn}
g^s_u ( \phi^s_u,  \gamma^s_u  ) = \E[ {f}^{s}_u ( \phi^{s}_{u}, \gamma^{s}_u D )],
\end{equation}
are jointly concave on $\Pi^{H,\delta}.$ 
\end{assumption}

\begin{lemma} \label{lm:condn:examplesg} Assumption \ref{condn:conc-g} is
  satisfied for systems where superposition/puncturing of each user is
  modelled via either a
\begin{enumerate}
\item Convex loss function or 
\item Threshold loss function with fixed relative thresholds,
  i.e., $t^s_u(\phi^s_u) = \alpha^s_u$ for $\phi \in [0,1]$ and the
  URLLC demand distribution $F_D(\cdot)$ is such that $F_D(\frac{1}{x})$ is
  concave in $x$ (satisfied by the truncated Pareto distribution).
\end{enumerate}
\end{lemma}
The proof is included in Appendix~C.
%
%
%
%
%
With this condition in place, we now describe the throughput region.

\begin{theorem} \label{thm:capacity-theorem} Under 
  Assumption~\ref{condn:conc-g} we have that
  ${\cal C}^{H,\delta}= \hat{\cal{C}}^{H,\delta}$. 
\end{theorem}
The proof is available in the  Appendix~D. The above theorem implies that we do not have to consider time-sharing/randomization amongst minislot-homogeneous joint scheduling policies.
 Thus, with minislot-homogeneous policies and under the concavity of $ g_u^{\pi, s}\brac{\cdot,\cdot} $
from Assumption~\ref{condn:conc-g}, the above result sets up a convex
optimization problem in $(\bvecgreek{\phi}, \bvecgreek{\gamma)},$ i..e, we have a concave cost
function with convex constraints. Thus, by iteratively updating
$(\bvecgreek{\phi}, \bvecgreek{\gamma)},$ we can develop an online scheduling algorithm that
asymptotically maximizes eMBB users' utility. This is descried next.

\subsection{Stochastic approximation based online algorithm}
\label{sec:th-stoch-approx}



We first restate the utility maximization problem for minislot-homogeneous
URLLC/eMBB scheduling policies:
\begin{align}
\underset{\bvec{\phi}, \bvec{\gamma} \in \Pi^{U,\delta}}{\max}  \quad \sum_{u \in
  \mathcal{U}}& U_u\brac{\sum_{s \in \mathcal{S}} p_{\mathcal{S}}
                (s)g^s_u \brac{ \phi^s_u,  \gamma^s_u } }.
\end{align}
Observe that the objective function is concave because it consists of a sum of compositions
of non-decreasing concave functions ($U_u(\cdot)$), and  concave functions
($g^s_u \brac{ \cdot, \cdot }$) in $\bvec{\phi}$ and
$\bvec{\gamma}$ (if Assumption~\ref{condn:conc-g} holds). Further, the constraint set is convex. Therefore, the
above problem fits in the framework of standard convex optimization
problems. However, solving the above problem requires knowledge of
all possible network states and their probability distribution,
resulting in an {\em offline} optimization problem.  In this section,
we develop a stochastic approximation based online algorithm to solve
the above problem.

{\bf Online algorithm: }
Let $\bvec{\overline{R}}(t-1):=\brac{\overline{R}_1(t-1), \overline{R}_2(t-1), \ldots, \overline{R}_u(t-1), \ldots, \overline{R}_{\abs{\mathcal{U}}}(t-1)}$ be the random vector denoting the average
rates received by eMBB users up to eMBB slot $t-1$ under our online algorithm. Let $ \overline{\bvec{r}}(t-1) $ denote a realization of $  \bvec{\overline{R}}(t-1)$.  Let $s$ be the
network state in slot $t$. Define vectors $\bvecgreek{\phi}^s:=\brac{\phi^s_u,
  \mid u \in \mathcal{U}}$ and $\bvecgreek{\gamma}^s:=\brac{\gamma^s_u \mid u \in
  \mathcal{U} }$.  At the beginning of eMBB slot $t$, we compute
 vectors $\brac{\bvecgreek{{\tilde{\phi}}} (t), \bvecgreek{{\tilde{\gamma}}} (t) }$ as the
solution to the following optimization problem: 
\begin{align}
    \underset{{\phi^s}, {\gamma^s}}{\max} \quad \sum_{u \in \mathcal{U}}&
                                                                      U_u^{'}\brac{\overline{r}_u(t-1)}
                                                                          g^s_u
                                                                          (
                                                                          \phi^s_u,
                                                                          \gamma^s_u
                                                                          )
                                                                          , \label{eqn:opt-stoch-approx}\\ 
\text{s.t.} \quad 
   \bvecgreek{\phi}^s & \geq \brac{1-\delta}\bvecgreek{\gamma}^s, \\
\sum_{u \in \mathcal{U}} \phi^s_u&=1  \mbox{  and  }  \sum_{u \in \mathcal{U}} \gamma^s_u=1, \\
\bvecgreek{\phi}^s &\in \sbrac{0,1}^{\abs{\mathcal{U}}} \mbox{  and  } \bvecgreek{\gamma}^s \in \sbrac{0,1}^{\abs{\mathcal{U}}}.
\end{align}
This optimization problem is a convex optimization problem and can be
solved numerically using standard convex optimization techniques.   
Using $\brac{\bvecgreek{{\tilde{\phi}}} (t), \bvecgreek{{\tilde{\gamma}}} (t) }$, we schedule
URLLC and eMBB traffic as follows: 

{\bf The eMBB scheduler:} For notational ease, we fluidize the
bandwidth.  Specifically, we assume that the bandwidth of a resource
block is very small when compared to the total bandwidth available.
Hence, the bandwidth can be split into arbitrary fractions and we
allocate fraction $\tilde{\phi}_u (t)$  of the total bandwidth to eMBB
user $u$.


{\bf The URLLC Scheduler:} We load different eMBB users with
URLLC traffic according to  the vector ${\bvecgreek{\tilde{\gamma}}}(t)$.   

At the end of eMBB slot $t$, the eMBB scheduler receives feedback from
the eMBB receivers indicating the rates received by the eMBB
users. Let us denote the rate received eMBB user $u$ in the slot by the
random variable $R_u(t)$.  We update $\overline{R}_u(t)$ as
follows: 
\begin{equation}
\label{eq:update_for_rbar}
\overline{R}_u(t)= \brac{1-\epsilon_t}\overline{R}_u(t-1) + \epsilon_t R_u(t), 
\end{equation}
where $\cbrac{\epsilon_t \mid t=1,2,3, \ldots}$ is a sequence of
positive numbers which satisfy the following (standard) assumption: 
\begin{assumption} The averaging sequence $\{\epsilon_t \}$ satisfies:
\label{eq:condition_on_epsilon}
$$\sum_{t=1}^{\infty} \epsilon_t = \infty \quad \mbox{and} \quad
\sum_{t=1}^{\infty} \epsilon_t^2 < \infty.$$ 
\end{assumption}

Finally, we state the main result of this section, which is the
optimality of the stochastic approximation based online algorithm. 
\begin{theorem}
\label{th:optimality_of_stoch_approx}
Let $\bvec{r}^*$ be the optimal average rate vector received by eMBB
users under the solution to the offline optimization problem. Suppose
that Assumptions~\ref{eq:condition_on_epsilon} and~\ref{condn:conc-g}  
hold, then we have that: 
\begin{equation}
 \underset{t \rightarrow \infty}{\lim} \bvec{\overline{R}}(t) = \bvec{r}^* \quad \text{almost surely}. 
 \end{equation} 

\end{theorem}
The proof is available in the Appendix~E. 

\subsection{Optimality of Minislot-Homogeneous Policies}
 In the previous section we restricted ourselves to  minislot-homogeneous policies. In this section will justify this choice. Let us consider a  generalization of minislot-homogeneous policies where  the URLLC placement in each minislot can depend on the history of URLLC arrivals prior to that minislot. Such a policy will obviously perform better than minislot-homogeneous URLLC placement policies since in a minislot-homogeneous policy we decide the URLLC placement at the beginning of an eMBB slot  based on the expected loss due to puncturing/superposition and do not adapt it based on the realization of URLLC demands per minislot. However, finding an optimal scheduling policy under this generalization can be computationally expensive as compared to minislot-homogeneous policies which are attractive due to their simplicity.  In this section we identify conditions under which minislot-homogeneous URLLC placement polices perform as well as the general class of  \emph{causal}  and \emph{minislot-dependent} policies.  These terms are defined below.   
\begin{definition}
	A scheduler is said to be   \emph{causal} if at the beginning of a mini-slot $m$  the scheduler knows the realizations of $D(1)$, $D(2), \ldots, D(m-1)$ and is unaware of the realizations of $ D(m), D(m+1), \ldots, D(\abs{\mathcal{M}}) $.
\end{definition}


\begin{definition}
	A scheduling policy is said to be \emph{minislot-dependent} if the URLLC placement policy can vary with the minislot index $m$ and previous URLLC demands in the eMBB slot.  
\end{definition}

The decision variables in  a   causal and  minislot-dependent joint scheduling  policy $\pi$ can be described as follows: 
\begin{enumerate}
	\item At the beginning of an eMBB slot, the  scheduler chooses $\phi_u^{\pi, s}, u \in \mathcal{U}$ such that
	\begin{equation}
	\label{eq:constraint_on_phi}
	\sum_{u \in \mathcal{U}} \phi^{\pi, s}_u =1  \mbox{  and  } \phi^{\pi, s}_u \in \sbrac{0,1} \quad \forall u \in \mathcal{U}. 
	\end{equation} 
	\item In each mini-slot $m$, the total puncturing placed on eMBB user $u$ is given by $\gamma_{u,m}^{\pi, s}\brac{ \bvec{d}^{\brac{1:m-1}}} D_m$, where $\gamma_{u,m}^{\pi, s}\brac{ \cdot}$ characterizes the URLLC placement in minislot $ m $ as function of the previously seen URLLC demands   $\bvec{D}^{\brac{1:m-1}}:= \brac{D(1), D(2), \ldots, D(m-1)}$. Let $ \bvec{d}^{\brac{1:m-1}} $ is a realization of $ \bvec{D}^{\brac{1:m-1}}$. For any $m$ and $\bvec{d}^{\brac{1:m-1}}$, $\gamma_{u,m}^{\pi, s}\brac{ \bvec{d}^{\brac{1:m-1}}}$  has to satisfy the following constraints. 
	\begin{align}
	\sum_{u \in \mathcal{U}} \gamma^{s,\pi }_{u,m} ( \bvec{d}^{\brac{1:m-1}})&=1,   \label{eq:constraint_on_gamma_1}\\
	\gamma^{s,\pi }_{u,m} ( \bvec{d}^{\brac{1:m-1}}) &\leq \frac{ \phi_u^{\pi, s}}{\abs{\mathcal{M}}\brac{1-\delta}} \quad \forall u \in \mathcal{U}, \label{eq:constraint_on_gamma_2}\\
	\gamma^{s,\pi }_{u,m} ( \bvec{d}^{\brac{1:m-1}}) &\in \sbrac{0,1} \quad \forall u \in \mathcal{U}. \label{eq:constraint_on_gamma_3}
	\end{align}  
\end{enumerate}
 Observe that the URLLC placement factor for causal and minislot-dependent  scheduling  policy is not just dependent on the user and network state but it also depends on the mini-slot index and  past URLLC demands. 
  Let $\tilde{\Pi}$ be the set of all  causal and mini-slot dependent  scheduling  policies. In our online algorithm~\eqref{eqn:opt-stoch-approx}, for any eMBB slot $t$, we  find the policy which solves the following optimization problem with non-negative weights $ w_u  $. 
\begin{equation}
\mathcal{OP}_1 : \quad
\underset{\pi \in \tilde{\Pi}}{\max}: \sum_{u \in \mathcal{U}} w_u g_u^{\pi, s}\brac{\phi_u^{\pi, s}, \bvecgreek{\gamma}_{u}^{\pi, s}},
\end{equation} 
where $s$ is the current network state, $ \bvecgreek{\gamma}_{u}^{\pi, s}:=\brac{{\gamma}_{u,1}^{\pi, s}\brac{ \cdot}, {\gamma}_{u,2}^{\pi, s}\brac{ \cdot}, \ldots, {\gamma}_{u,\abs{\mathcal{M}}}^{\pi, s}\brac{ \cdot}} $ is the vector of URLLC placement factors of all minislots (with slight abuse of notation)  and $g_u^{\pi, s}\brac{\cdot, \cdot}$ is the average rate experienced by eMBB user $u$ under policy $\pi$. $g_u^{\pi, s}\brac{\cdot, \cdot}$ is given by the following expression:  
 \begin{multline}
 g_u^{\pi, s} \brac{\phi_{u}^{\pi, s}, \bvecgreek{\gamma}_{u}^{\pi, s}} := \\ r_u^s \phi_u^{\pi, s}\expect{{1-   h_u^s \brac{\frac{\sum_{m=1}^{\abs{\mathcal{M}}}{\gamma_{u,m}^{\pi, s}\brac{{ \bvec{D}^{\brac{1:m-1}}} }D_m}}{{ \phi_u^{\pi, s}}}}}},
 \end{multline}
 where the expectation is computed with respect to the joint distribution of $D(1)$, $D(2), \ldots$, $D(\abs{\mathcal{M}})$. One can formulate the above optimization problem as a Markov Decision Problem (MDP), however the state space for such an MDP is prohibitively large.  Furthermore we note that minislot-homogeneous policies are attractive in terms of it computational complexity. In general, one cannot expect optimal minislot-homogeneous policies to perform as well as  optimal minislot dependent policies, however, if we restrict ourselves to convex \emph{homogeneous} loss functions, then we can show that minislot-homogeneous policies are in fact optimal over $ \tilde{\Pi} $.

 \begin{definition}
 	\label{asm:assumption_on_convex_cost_functions}
 	A loss function $ h_u^s(\cdot) $ is said to be homogeneous if there exists a  real number $p$ such that $ \forall $ $ x \in \sbrac{0,1} $ and $ \kappa \geq 0$ we have that  
 	\begin{equation}
 	h_u^s(\kappa x) = \kappa^p h_u^s( x).
 	\end{equation}
 \end{definition} 
 
 Even with this restriction we can model useful loss functions which could possibly be user and network state dependent. Some examples are given below. 
 \begin{enumerate}
 	\item {\bf Linear:} $h_u^s(x) = k_u^s \brac{x}$, where $k_u^s\geq 0$.
 	\item {\bf Monomial:} $h_u^s(x)= k_u^s \brac{x}^q$ where $k_u^s\geq 0$ and $q \geq 1$. 
 \end{enumerate}
 
 Our main result on the optimality of minislot-homogeneous policies  is proved in Appendix~F  and stated next. 
 \begin{theorem}
 	\label{thm:optimality_of_minislot_homogeneous_policy}
If the support of URLLC demands $ D $ is a finite discrete set and eMBB loss functions are homogeneous and convex, then there exists an optimal solution  $\brac{\bvec{\phi}^{s, *}, \bvec{\gamma}^{s, *}\brac{ \cdot}}$ for $\mathcal{OP}_1$ with a minislot-homogeneous URLLC placement policy $ \bvec{\gamma}^{s, *} $. 
\end{theorem}



\subsection{Optimal eMBB Slot Slicing}
\label{sec:eMBB_slot_aggregation}
In Section~\ref{sec:sys-model} we have used uniform slot sizes for eMBB users, i.e. the allocated minislots to all users span the entire width of the slot (see Figure~\ref{fig:config_2}; henceforth referred to as frequency slices). However, new proposals allow greater flexibility in slot allocation, e.g., the capability to choose different slices over both time and frequency for different eMBB users~\cite{3gpp_ran1_87}. 
In this section we will show that while it is possible to slice eMBB users' resources flexibly, it is preferable to slice frequency (see~\ref{fig:config_2}) than time from the point of view of puncturing losses for convex loss functions.  

 \begin{figure}
\centering
\includegraphics[scale=0.33]{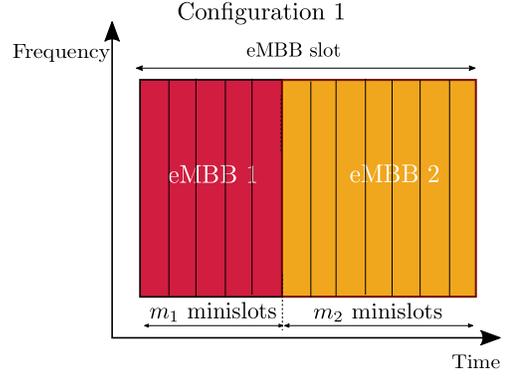}
\caption{Time Slices: In this configuration, eMBB users share  resources over time in an eMBB slot.}
\label{fig:config_1}
\end{figure}

\begin{figure}
\centering
\includegraphics[scale=0.35]{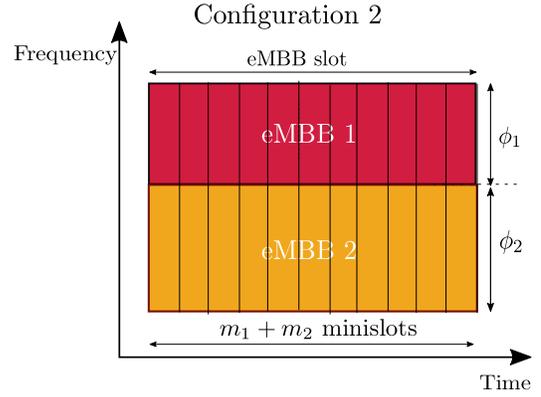}
\caption{Frequency Slices: In this configuration, eMBB users homogeneously share frequency in an eMBB slot.}
\label{fig:config_2}
\end{figure}

The essence of the discussion can be captured by comparing the two resource allocation configurations  shown in Figures~\ref{fig:config_1} and~\ref{fig:config_2}. In Configuration~1 (time slices),  eMBB user 1 is allocated the entire frequency band for a subset of $ m_1 $ minislots.
Similarly  eMBB user 2 is allocated the entire frequency band for its subset of $ m_2 $ minislots.  The network state $s$ is assumed to be the same for the entire $ m_1 + m_2$ minislots. This implies that the loss functions of eMBB users $ (h_u^s\brac{\cdot}) $ do not change throughout the $ m_1 + m_2 $ minislots. In Configuration~2 (frequency slices) we allocate an eMBB user $1 $ a fraction $\phi_1$ of the bandwidth for a duration of $ m_1 + m_2 $ minislots, where $ \phi_1:=\frac{m_1}{m_1 +m_2}$ and similarly for eMBB user $ 2 $.  Note that the total resources allocated to eMBB users, which is represented by the area allocated  in the time-frequency plane is same in both configurations.

In Configuration~1, the total puncturing observed by eMBB user $1$ is given by $\sum_{m=1}^{m_1} D(m) $ and similarly for eMBB user $ 2 $.  Whereas in Configuration~2, under uniform URLLC placement, the total puncturing observed by eMBB user $1$ is given by $\sum_{m=1}^{m_1+m_2} \phi_1 D(m)$. Note that the mean total puncturing is same in both the configurations. 

The main result of this section is given below:
\begin{theorem}
	\label{thm:horizontal_vs_vertical}
Under the assumption of  i.i.d.  URLLC demands\footnote{This result can be extended to exchangeable URLLC demands. We use i.i.d. assumption to maintain consistency with other sections. } ($D(m), \, m=1,2, \ldots, m_1+m_2$)  and convex loss functions $(h_u^s \brac{\cdot})$,   for any eMBB user, e.g., eMBB user $ 1 $, we have that
\begin{equation}
\label{eq:result_on_vertical_vs_horizontal}
 \expect{h_1^s \brac{\sum_{m=1}^{m_1} D(m)}}  \geq \expect{h_1^s \brac{\sum_{m=1}^{m_1+m_2 } \phi_1 D(m)}}.  
 \end{equation} 
\end{theorem}
Proof of this result is given in Appendix~G.

{\bf Remarks:} The above theorem shows that the expected  loss suffered by an eMBB user due to URLLC puncturing in  Configuration~1 (time slicing) is higher than in Configuration~2 (frequency slicing).
This implies that it is preferable for eMBB users to spread their resource allocation over time from the perspective of reducing their loss due to puncturing.  The underlying reason is that Configuration~2 results in smaller variability in the total puncturing even though both the configurations have the same mean total puncturing. Since the loss functions are convex, a lower variability leads to a lower expected loss.  Finally, for more complex (rectangular) slices, we can now apply Thm.~\ref{thm:horizontal_vs_vertical} iteratively and show that using frequency slices with appropriate scaling of the bandwidth allocation results in a higher average rate for eMBB users.

%


\section{Threshold Model and Placement Policies}
\label{sec:threshold_model}

In the previous section, we developed a stochastic approximation based
algorithm for minislot-homogeneous policies. This algorithm
iteratively solves the optimization problem given in
(\ref{eqn:opt-stoch-approx}). This optimization problem jointly
optimizes over a pair of row vectors $(\bvecgreek{\phi}^s, \bvecgreek{\gamma}^s).$ While
this convex optimization problem can be solved using standard methods,
it could become computationally challenging as the number of users
increases.

In this section, we shall restrict our attention to a threshold model
for superposition/puncturing, and look at policies that impose
structural conditions on the puncturing matrix $\bvecgreek{\gamma}.$ We
will show that the resulting class of policies have nice theoretical
properties that lead to simpler online algorithms (solving
(\ref{eqn:opt-search}), which is an one-dimensional search).
%



We consider two types of structural conditions on  $\bvec{\gamma}$:

\noindent \textbf{(i) Resource Proportional (RP) Placement:} The first
is based on allocating URLLC demands in proportion to eMBB user slot
allocations, i.e., $\gamma^s_{u} = \phi^s_{u}.$ We refer to this as
Resource Proportional (RP) Placement and denote such policies by
$$
\Pi^{RP,\delta} := \{ ({\bm \phi},{\bm \gamma}) \in  \Pi^{H, \delta} ~|~ 
 {\bm \gamma} ={\bm \phi} \},
$$
and define the associated achievable throughput region 
\begin{eqnarray*} {\cal C}^{RP,\delta} &=& \{ \mathbf{c} \in
  \Reals_+^{|{\cal U}|} ~|~ \exists {\bm \pi} \in \Pi^{RP,\delta}
  ~\mbox{s.t.}~ \mathbf{c} \leq {\bf c}^{{\bm \pi}} \}.
\end{eqnarray*}
The motivation for RP Placement comes from the optimality of random
placement for the linear model in
Section~\ref{sec:erasure_channel}. Observe that if puncturing occurs
uniformly randomly, then the expected number of punctures is directly
proportional to the fraction of bandwidth allocated to an eMBB
user. Thus, RP Placement can be viewed as a {\em determinized
  version} of the random placement strategy which ensures that the proportions of puncturing satisfy  resource proportional ratios. 

\noindent \textbf{(ii) Threshold Proportional (TP) Placement:} The
second policy allocates URLLC demands in proportion to the eMBB users
associated loss thresholds so as to avoid losses,
$$
\gamma^s_{u} =  \frac{\phi^s_{u} t^s_{u}(\phi^s_{u})}{\sum_{u' \in {\cal U}} \phi^s_{u'} t^s_{u'}(\phi^s_{u'})}.
$$   
We refer to this as Threshold Proportional (TP) Placement and denote such policies by
\begin{eqnarray*}
    \Pi^{TP,\delta}  := && \\
&&\hspace*{-50pt} \{ ({\bm \phi},{\bm \gamma}) \in \Pi^{H, \delta} ~|~ \gamma^s_{u}
     =  \frac{\phi^s_{u} t^s_{u}(\phi^s_{u})}{\sum_{u' \in {\cal U}}
     \phi^s_{u'} t^s_{u'}(\phi^s_{u'})} 
\forall s \in {\cal S}, u \in {\cal U} \}.
\end{eqnarray*}
The associated achievable throughput region is denoted  
\begin{eqnarray*}
{\cal C}^{TP,\delta}  & = & \{ \mathbf{c} \in \Reals_+^{|{\cal U}|}
~|~ \exists {\bm \pi} \in \Pi^{TP,\delta} ~\mbox{s.t.}~ \mathbf{c} \leq {\bf c}^{{\bm \pi}} \}.  
\end{eqnarray*}

First we state  a corollary to Theorem~\ref{thm:main-theorem} which characterizes the rates under different URLLC placement policies for systems having threshold loss model for superposition/puncturing.

\begin{corollary} \label{cor:main-theorem-th} Under a $(1-\delta)$
	sharing factor and time-homogeneous scheduler
	${\bm \pi} = ({\bm \phi}^{\bm \pi},{\bm \gamma}^{\bm \pi}) \in
	\Pi^{H,\delta}$ the probability of induced eMBB loss for user
	$u\in {\cal U}$ in channel state $s \in {\cal S}$ is given by
	$$
	\epsilon^{{\bm \pi},s}_u = 1-F_ D( \frac{\phi^{{\bm \pi},s}_u
		t^{s}_u(\phi^{{\bm \pi},s}_u)}{\gamma^{{\bm \pi},s}_u} ).
	$$
	where $F_D$ denotes the cumulative distribution function of the URLLC
	demands on a typical eMBB slot. 
	Then the associated user throughput is given by
	$$
	r^{{\bm \pi},s}_{u} = \hat{r}^{s}_u \phi^{{\bm \pi},s}_{u} F_
	D(\frac{\phi^{{\bm \pi},s}_u t^{s}_u(\phi^{{\bm
				\pi},s}_u)}{\gamma^{{\bm \pi},s}_u} ).
	$$
	and the overall user throughputs are given by ${\bf c}^{{\bm \pi}}= (
	c^{{\bm \pi}}_{u} : u \in {\cal U})$ where 
	$$
	c^{{\bm \pi}}_{u} = \sum_{u \in {\cal U}} \hat{r}^{s}_u \phi^{s}_{u}
	F_ D(\frac{\phi^{{\bm \pi},s}_u t^{s}_u(\phi^{{\bm
				\pi},s}_u)}{\gamma^{{\bm \pi},s}_u} ) p_S(s).
	$$
\end{corollary}

The following two corollaries are direct consequences of Corollary
\ref{cor:main-theorem-th} and Theorem~\ref{thm:capacity-theorem}
restricted to RP and TP Placement strategies, and characterize the
capacity regions under the two policies.

\begin{corollary}
	\label{cor:rp} 
	Consider a wireless system with full sharing factor and
	time-homogeneous scheduler based on the RP URLLC Placement policy
	${\bm \pi} = ({\bm \phi}^{\bm \pi},{\bm \gamma}^{\bm \pi}) \in
	\Pi^{RP,\delta}.$
	Then any eMBB resource allocation $\bm{\phi}$ combined with a RP URLLC
	demand placement policy, $\bm{\gamma}= \bm{\phi}$ is feasible.  The
	probability of loss for user $u\in {\cal U}$ in channel state
	$s \in {\cal S}$ is given by
	$$
	\epsilon^{{\bm \pi},s}_u = 1-F_ D(
	{t^{s}_u(\phi^{{\bm \pi},s}_u)} ),
	$$
	with associated user throughput 
	\begin{eqnarray}
	r^{{\bm \pi},s}_{u} &=&  \hat{r}^{s}_u \phi^{s}_{u} F_
	D({t^{s}_u(\phi^{{\bm \pi},s}_u)}
	). \label{eqn:RP-cor-rate}
	\end{eqnarray}
	Further if for all $s\in {\cal S}$ and $u \in {\cal U}$
	the functions $g^s_u ( , ) $ given by  
	\begin{eqnarray}
	g^s_u ( \phi^s_u ) &=&  \phi^s_u  F_
	D({t^{s}_u(\phi^{{\bm \pi},s}_u)} ), \label{eqn:RP-cor-g}
	\end{eqnarray}
	are concave then ${\cal C}^{RP,\delta}= \hat{\cal C}^{RP,\delta}.$ 
\end{corollary}

\begin{corollary}
	\label{cor:tp}
	Under a  $(1-\delta)$ sharing factor and jointly uniform scheduler 
	based on the TP URLLC Placement policy ${\bm \pi} = ({\bm \phi}^{\bm
		\pi},{\bm \gamma}^{\bm \pi}) \in \Pi^{TP,\delta},$  
	the probability of induced eMBB loss
	user $u\in {\cal U}$ in channel state $s \in {\cal S}$ is given by  
	\begin{eqnarray}
	\epsilon^{{\bm \pi},s}_u &=& 1-F_ D( \sum_{u \in {\cal U}}
	\phi^{{\bm \pi},s}_u t^{s}_u(\phi^{{\bm \pi},s}_u) ), \label{eqn:tp-bound}
	\end{eqnarray}
	with associated user throughput 
	\begin{eqnarray}
	r^{{\bm \pi},s}_{u} =  \hat{r}^{s}_u \phi^{s}_{u} F_ D( \sum_{u \in {\cal
			U}} \phi^{{\bm \pi},s}_u t^{s}_u(\phi^{{\bm \pi},s}_u) ). \label{eqn:TP-cor-rate}
	\end{eqnarray}
	Further if for all $s\in {\cal S}$ and $u \in {\cal U}$
	the functions $g^s_u ( , ) $ given by  
	\begin{eqnarray}
	g^s_u ( \phi^s_u,  \gamma^s_u  ) &=&  \phi^s_u F_ D( \sum_{u \in
		{\cal U}} \phi^{{\bm \pi},s}_u
	t^{s}_u(\phi^{{\bm \pi},s}_u)
	), \label{eqn:TP-cor} 
	\end{eqnarray}
	are jointly concave then ${\cal C}^{TP,\delta}= \hat{\cal C}^{TP,\delta}.$ 
\end{corollary}

The following theorem provides a formal motivation for TP Placement. 
The main takeaway here is that the {\em probability of any loss in an
  eMBB slot under TP Placement policy is a lower bound for all other
  strategies.} Note that minimizing the probability of any eMBB loss is not same as minimizing eMBB rate loss.  


\begin{theorem} \label{thm:tp-optimality-theorem}
Consider a system with $(1-\delta)$ sharing factor. Consider
a joint scheduling policy based on the TP URLLC placement i.e, 
${\bm \pi} = ({\bm \phi}^{\bm \pi},{\bm \gamma}^{\bm \pi}) \in \Pi^{TP,\delta}.$ 
Then ${\bm \pi}$ achieves the minimum probability of any eMBB loss amongst
all joint scheduling policies using the same eMBB resource allocation ${\bm \phi}^{\bm \pi}.$
\end{theorem}

The proof is  included in Appendix~H. 
Next  we consider online algorithms that implement the RP
and TP Placement policies. While the stochastic approximation
algorithm developed in Section~\ref{sec:th-stoch-approx} can clearly
be used, the additional structure imposed by the RP and TP Placement
policies, and the shape of the threshold loss function (discussed
below) can result in much simpler algorithms (with optimality
guarantees).

We consider the case where $t^s_u(\phi)$ is a (state dependent but
$\phi$ independent) constant, i.e., $t^s_u(\phi) = \alpha^s,$
where $\alpha^s \in (0,1).$ Intuitively, this means that eMBB traffic
which has a higher share of the bandwidth is more resilient to losses
(e.g. through coding over larger fraction of resources).
Then, by substituting this loss function in (\ref{eqn:RP-cor-rate}) and
(\ref{eqn:TP-cor-rate}) (where we also use the fact that $\sum_{u
  \in  {\cal U}} \phi^{s}_u = 1$), we have that 
$$
r^{{\bm \pi},s}_{u} =  \hat{r}^{s}_u \phi^{s}_{u} F_D(\alpha^s).
$$
Comparing with the development in Section~\ref{sec:util-linear}, we
observe that the cost and constraints are identical if
$F_D(\alpha^s)$ replaces $(1 - \rho).$ Note that a small difference
is that $F_D(\alpha^s)$ is state dependent, whereas
$(1 - \rho)$ does not depend  on the state; however, it is easy to see
that the development in Section~\ref{sec:util-linear} immediately
generalizes to this setting. Hence, we can interpret
$F_D(\alpha^s)$ as the state dependent average rate loss
due to puncturing via the RP or TP Placement policies.

We can now employ the rate-based iterative gradient scheduler
developed in Section~\ref{sec:util-linear} (by replacing $(1 - \rho)$
in (\ref{eqn:lin-rate-update}) by a user-dependent $F_D(\alpha^s)$),
and the theoretical guarantees directly carry over. As this algorithm
only minimizes over users at each slot in (\ref{eqn:opt-search}), this
is easier to implement when compared to the stochastic approximation
algorithm developed in Section~\ref{sec:th-stoch-approx}.

\section{Simulations}
We consider a system with a total of 100 RBs available per eMBB slot, and
with 8 minislots per eMBB slot.
In an eMBB slot, $\hat{r}_u^s$ for an eMBB user is drawn from the
finite set $\cbrac{1, 2, 3, 4, 5, 6, 7, 8, 9, 10}$
Mbps  according to a probability distribution and i.i.d. across users and slots.  Our
system consists of 20 users, and with 100 channel states (all equally
likely).  The ($20$ users $\times$ $100$ states) rate matrix is
one-time synthesized by independently and uniformly sampling a rate
from the finite rate set for each matrix element. For $ 10$ eMBB users, we have chosen the probability distribution such that the average rate is $ 7 $ Mbps. For the rest, probability distribution is such that the average rate is $ 3 $ Mbps. This models two classes of users, one class with higher link rates which can tolerate a higher amount of puncturing and the other with lower link rates which can tolerate lesser amount of puncturing. This is reasonable as a user with a higher channel rate can code more robustly and protect its transmissions from URLLC puncturing  more than a user with  a lower channel rate. In this spirit we shall call users with $ 7 $ Mbps average rates as `robust' users and users with $ 3 $ Mbps average rates as `sensitive' users. We use the utility function $ U_u(r)=\log \brac{r} $ for all users.

We first show that joint scheduling is necessary to preserve eMBB throughputs. To that end we benchmark  our optimal online algorithm (stochastic approximation algorithm, see
Section~\ref{sec:th-stoch-approx}) for convex loss functions with a scheme which performs standard gradient based scheduling for eMBB users and Resource Proportional (RP)  URLLC placement. Note that for convex loss functions, RP placement strategy does not take into account the  eMBB user's sensitivity to delays.  For  users with average rate $7$ Mbps, we use the loss function $ h_u^s(x)=x^2 $. For users with average rate $ 3 $ Mpbs, we use the following loss function:
\begin{equation}
h_u^s(x)=
	\begin{cases}
	\brac{\frac{x}{0.7}}^2, & \mbox{if } x \leq 0.7, \\
	0,  & \mbox{if } 0.7< x \leq 1.
	\end{cases}
\end{equation}

URLLC demands in a minislot is drawn from a binomial distribution which can take values $ 0 $ with $ p $ and $ \frac{1-\delta}{8} $ with probability $ 1-p $. Note that this ensures that peak URLLC load in an eMBB slot is less than or equal to $ 1-\delta $.   

In Fig.~\ref{fig:convex_loss_utility}, we  compare the average sum utility under our optimal joint scheduler and the RP based policy as a function of the URLLC load. As the load increases,  RP performs poorly. To understand this phenomenon in detail, we have plotted the average rates of robust and sensitive users under the two policies in Fig.~\ref{fig:convex_loss_avg_rate}. As we increase the URLLC load, the average eMBB rates of both sensitive and robust users decrease rapidly. For example, when $ \rho=0.4 $, RP has $15$  \% lower throughput for robust users and almost similar performance for sensitive users as compared to optimal algorithm. Further as we increase $ \rho $ to $ 0.6 $, the throughput of robust and sensitive users in RP decrease by $ 35 $ \% and $ 26 $ \%, respectively. 

 Sensitive users are the most affected by URLLC puncturing. When the RP URLLC placement policy is combined with the standard gradient based algorithm for eMBB users,  it allocates more resources to sensitive users because they have higher marginal utility. Since sensitive users receive more bandwidth, under the RP URLLC placement strategy they receive more puncturing. This will lead to even more allocation of resources to sensitive users and this process continues until robust users have similar marginal utilities (due to reduced rates) as sensitive users. Hence, the robust users are resource starved. As we increase the URLLC load further, sensitive users receive even more URLLC puncturing and neither the robust nor sensitive users get good average rates when compared to the optimal joint scheduler. This shows that we require joint scheduling of eMBB and URLLC to exploit the  heterogeneity in sensitivities to URLLC puncturing in maximizing eMBB utilities. 

\begin{figure}
	\centering
	\includegraphics[height=2in,width=3in]{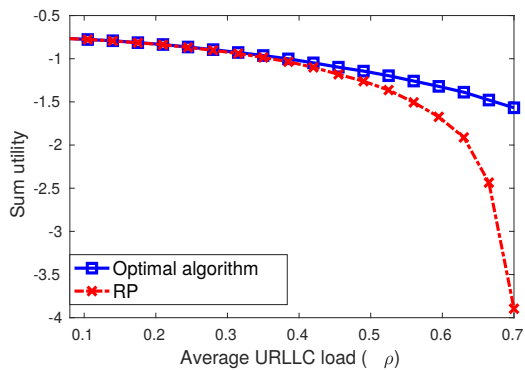}
	\caption{Sum utility as a function of URLLC load $\rho$ for the
		optimal and  RP  policies under convex model $ \delta=0.3 $.}
	\label{fig:convex_loss_utility}
\end{figure}

\begin{figure}
	\centering
	\includegraphics[height=2in,width=3in]{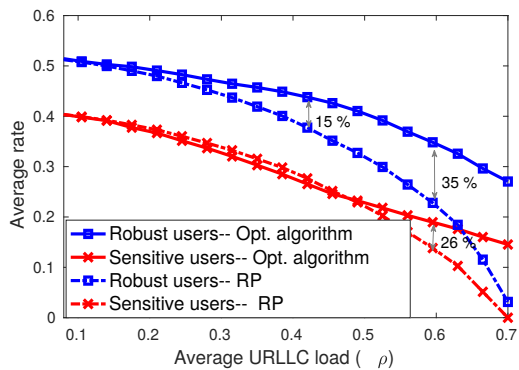}
	\caption{Average rates as a function of URLLC load $\rho$ for the
		optimal and  RP  policies under convex model $ \delta=0.3 $.}
	\label{fig:convex_loss_avg_rate}
\end{figure}

\begin{figure}
\centering
\includegraphics[height=2in,width=3in]{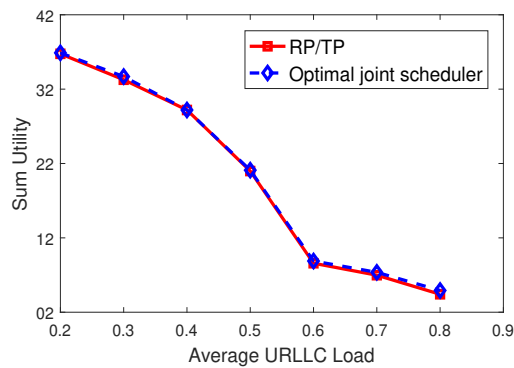}
\caption{Sum utility as a function of URLLC load $\rho$ for the
  optimal and TP Placement policies under threshold model
  ($\delta = 0.1$).}
\label{fig:threshold_model}
\end{figure}





Next we  consider a threshold based loss model with $\alpha^s =0.3$ for 50\% of
eMBB states and $\alpha^s =0.7$ for the rest. We use the utility
function $U_u(r)=\log(r) + 6.5$ for all eMBB users, where $r$ is
measured in Mbps (constant added to ensure non-negativity of the sum
utility). URLLC load in an eMBB slot ($D$) is generated based on
the truncated Pareto distribution with tail exponent $\eta=2$.
We compare the optimal policy (stochastic approximation algorithm, see
Section~\ref{sec:th-stoch-approx}) with that from the TP Placement
policy (the simpler gradient algorithm in
Section~\ref{sec:threshold_model}). In this case, since the threshold
functions are (state-dependent) constants, the RP and TP Placement
policies are the same. As we can see in
Figure~\ref{fig:threshold_model}, unlike the convex loss model the RP/TP Placement policy tracks the
optimal policy very well.

\begin{figure}
\centering
\includegraphics[height=2in,width=3in]{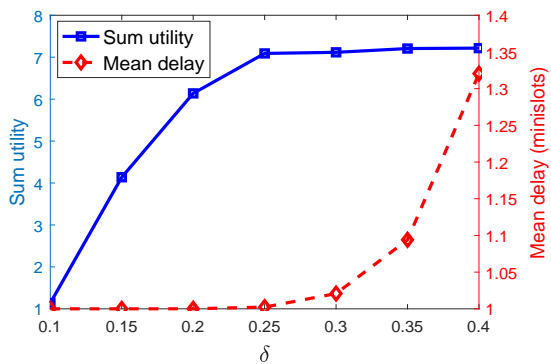}
\caption{Sum utility and mean URLLC delay  as a function of $\delta$.}
\label{fig:URLLC_cap_vs_emBB_util}
\end{figure}

\begin{figure}
\centering
\includegraphics[height=1.65in,width=3in]{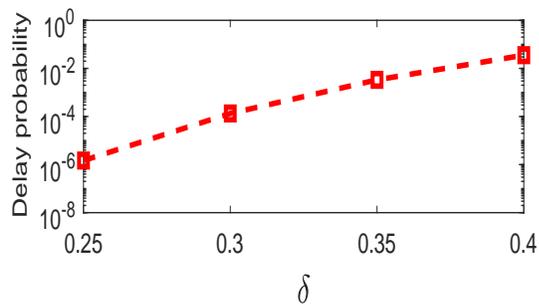}
\caption{Log-scale plot of the probability
  that URLLC traffic is delayed by more than two minislots (0.25 msec)
  for various values of $\delta.$ }
\label{fig:URLLC_cap_vs_emBB_util}
\end{figure}

In Figure~\ref{fig:URLLC_cap_vs_emBB_util}, we study the trade-off
between achieving a higher eMBB utility and lowering the mean delay of
URLLC traffic for different values of the sharing factor
$1-\delta$. Figure~\ref{fig:URLLC_cap_vs_emBB_util} plots the
corresponding probability that the URLLC traffic delay exceeds two
minislots ($0.125 \times 2 = 0.25$ msec).
To study this trade-off we generate URLLC arrivals in each minislot
from an uniform distribution between $\sbrac{0, 1/8}$ (recall there
are 8 minislots).
In each minislot, we can serve at most $\frac{1-\delta}{8}$ units of
URLLC traffic. If the URLLC load in a given minislot is more than
$\frac{1-\delta}{8}$, the remaining URLLC traffic is queued and served
in the next minislot on a FCFS basis.
For the eMBB users we use a convex model with
$h_u^s(x) = e^{{\kappa_u\brac{x-1}}}$ where $\kappa_u$ determines the
sensitivity of an eMBB user to an URLLC load. We have chosen
$\kappa=0.2$ for 50 \% of the users and $\kappa=0.7$ for the rest. We
also set $\forall u$ $U_u(x)=\log(x) + 4.2$ (constant added to ensure
positive sum utility).
In summary, a larger value of $\delta$ limits the amount of URLLC
traffic than can be served in a minislot. However, a larger $\delta$
enlarges the constraint set $\Pi^{H,\delta}$ in the eMBB utility
maximization problem, and hence we get higher eMBB utility.


\section{Conclusion}
\label{sec:concl}

In this paper, we have developed a framework and algorithms for joint
scheduling of URLLC (low latency) and eMBB (broadband) traffic in
emerging 5G systems. Our setting considers recent proposals where
URLLC traffic is dynamically multiplexed through
puncturing/superposition of eMBB traffic. Our results show that this
joint problem has structural properties that enable clean
decompositions, and corresponding algorithms with theoretical
guarantees.

\section*{Acknowledgements}
The work of Arjun Anand was partially supported by FutureWei Technologies and NSF grant CNS-1731658,
Gustavo de Veciana was partially supported by NSF grants CNS-1343383 and CNS-1731658, and Sanjay Shakkottai was partially
supported by NSF grants CNS-1343383 and CNS-1731658, and the US DoT D-STOP Tier 1
University Transportation Center.

\bibliographystyle{IEEEtran}
\bibliography{Bibtex/bibJournalList,Bibtex/arjun,Bibtex/ss-3}
\index{Bibliography@\emph{Bibliography}}

\newpage
\appendix

\subsection{Proof of Theorem~\ref{thm:theorem_on_erasure_channel_based_model}}
\label{pf:theorem_on_erasure_channel_based_model}
Clearly since $\Pi^{LR} \subset \Pi$ we have that ${\cal C}^{LR} \subset {\cal C}$ 

Now consider any policy $\pi \in \Pi$ with 
eMBB user allocations ${\bm \phi}^{\pi}$ and URLLC loads 
$\overline{{\bf l}}^{\pi}$ and associated long term throughput is $\mathbf{c}^{\pi}$ given by
$$
c_u^\pi = \sum_{s\in {\cal S}} \hat{r}_{u}^{s}( \phi^{\pi,s}_{u} - \overline{l}_{u}^{\pi,s})  p_S(s).
$$

Let us define a $\pi'$ based on $\pi$ to have per minislot eMBB user allocations given by  
$$
\phi_{u,m}^{\pi',s} = 
\frac{\phi_{u}^{\pi,s}- \overline{l}_{u}^{\pi,s}} {\sum_{u' \in {\cal U}} \phi_{u'}^{\pi,s}- l_{u'}^{\pi,s}} f
= \frac{\phi_{u}^{s}- \overline{l}_{u}^{\pi,s}}{ 1- \rho} f,
$$
for $s \in {\cal S},$ $u \in {\cal U}$ and $m \in {\cal M}.$ 
Since induced mean loads on an eMBB user can not exceed its allocation we have that
${\bm \phi}^\pi  \geq  \overline{{\bf l}}^\pi$ so the above allocations are positive. Note
also that this allocation is not minislot dependent, but normalized so that per minislot
they sum to $f$ and over the whole eMBB slot sum to $1$, i.e., ${\bm \phi}^{\pi'} \in \Sigma$.
Thus for such an allocation we have that 
$$
\phi_{u}^{\pi',s} = \frac{\phi_{u}^{s}-\overline{l}_{u}^{\pi,s}}{ 1- \rho}.
$$
Also suppose that $\pi'$ uses randomized URLLC placement across minislots which induces
mean URLLC loads proportional to the allocations, i.e.,  
$\overline{l}^{\pi',s}_{u} = \rho  \phi^{\pi',s}_{u}$. 
It follows that
\begin{eqnarray*}
\phi^{\pi',s}_{u}- \overline{l}^{\pi',s}_{u}  
& = & \phi^{\pi',s}_{u} - \rho \phi^{\pi',s}_{u} \\
& = & (1-\rho)  \phi^{\pi',s}_{u}   \\
&  =&    \phi_{u}^{\pi,s}-\overline{l}^{\pi,s}_{u},
\end{eqnarray*}
and so $c_u^{\pi,s}= c_u^{\pi',s}$ for all $s \in {\cal S}$ and $u \in {\cal U}.$ 
Thus for any policy $\pi$ there is a policy $\pi'$ which
uses randomized URLLC placement and achieves the same long term throughputs. 
It follows that ${\cal C} \subset {\cal C}^{LR}$ and so $ {\cal C} =  {\cal C}^{LR}.$ 

\subsection{Proof of Theorem~\ref{thm:main-theorem}}
\label{pf:main-theorem}

  Under a policy
  ${\bm \pi} = ({\bm \phi}^{\bm \pi},{\bm \gamma}^{\bm \pi}) \in
  \Pi^{H,\delta}$ we have that the induced loads are given by
$$
L_{u,m}^{{\bm \pi},s} =  \frac{\gamma_{u,1}^{{\bm \pi},s}}{f} D(m),
$$
so we have that
$$
L_{u}^{{\bm \pi},s} 
=  \sum_{u \in {\cal U}} L_{u,m}^{{\bm \pi},s} 
=  \frac{\gamma_{u,1}^{{\bm \pi},s}}{f} \sum_{u \in {\cal U}} D(m) 
=  \frac{\gamma_{u,1}^{{\bm \pi},s}}{f} D
=  \gamma_{u}^{{\bm \pi},s} D.
$$
where the last equality follows from the uniformity of URLLC splits 
and normalization it follows that
$$
r^{{\bm \pi},s}_{u} = \E[ {f}^{s}_u ( \phi^{{\bm \pi},s}_{u}, L^{{\bm
    \pi},s}_u)] = \E[ {f}^{s}_u ( \phi^{{\bm \pi},s}_{u}, \gamma^{{\bm
    \pi},s}_u D )].
$$

\subsection{Proof of Lemma~\ref{lm:condn:examplesg}}
\label{pf:condn:examplesg}
Recall that convex loss functions are specified as follows
$$
f_{u}^{s}(\phi_{u}^{s},l_{u}^{s}) = \hat{r}_{u}^{s} \phi_{u}^{s} (1-
h^s_u \brac{ \frac{\l_{u}^{s}}{\phi_{u}^{s}}}),
$$
with $h^s_u: [0,1] \rightarrow [0,1]$ a convex increasing
function. For time-homogenous policies we have defined
\begin{eqnarray*}
g^s_u ( \phi^s_u,  \gamma^s_u  ) & = & \E[ {f}^{s}_u ( \phi^{s}_{u}, \gamma^{s}_u D )] \\
& = & \hat{r}_{u}^{s} E[ \phi_{u}^{s}  - \phi_{u}^{s}  h^s_u ( \frac{\gamma^{s}_u}{\phi^{s}_u} D )].
\end{eqnarray*}
Recall that convex function $h(\cdot)$ one can define a function
$ l(\phi, \gamma) = \phi h (\frac{\gamma}{\phi})$ known as the
perspective of $h(\cdot)$ which is known to be jointly convex in its
arguments.  It follows that $ \phi-\phi h (\frac{\gamma}{\phi})$ is
jointly concave, and so is $g^s_u(\cdot)$ since it is a weighted
aggregation of jointly concave functions.

For threshold-based loss functions where $t^s_u(\phi^s_u) = \alpha^s_u$  we have that
\begin{eqnarray*}
g^s_u ( \phi^s_u,  \gamma^s_u  ) & = & \E[ {f}^{s}_u ( \phi^{s}_{u}, \gamma^{s}_u D )] \\
& = & \hat{r}_{u}^{s} \phi^{\pi,s}_u P(\gamma^s_u D  \leq  \phi_u^{\pi,s} \alpha^u_s )\\
& = & \hat{r}_{u}^{s} \phi^{\pi,s}_u F_D ( \frac{\phi_u^{\pi,s} \alpha^u_s}{\gamma^s_u}).
\end{eqnarray*}
Now using the same result on the perspective functions of variables
the result follows.  The truncated Pareto case can be easily verified
by taking derivatives.


\subsection{Proof of Theorem~\ref{thm:capacity-theorem}}
\label{pf:capacity-theorem}
Clearly  ${\cal C}^{H,\delta} \subset {\cal \hat{C}}^{H,\delta}.$ 
We will show that 
${\bf c} \in {\cal \hat{C}}^{H,\delta}$ then their exists   
${\bm \pi} = ({\bm \phi}^{\bm \pi},{\bm \gamma}^{\bm \pi}) \in \Pi^{H,\delta}$ such that
${\bf c} \leq {\bf c}^{{\bm \pi}}$ from which it follows that 
${\cal C}^{H,\delta} \subset {\cal C}^{H,\delta}.$ 

Suppose ${\bf c} \in {\cal \hat{C}}^{H,\delta}$, then it can be
represented as a convex combination of policies $\Pi^{H,\delta}$, in
each channel state. For example suppose for simplicity that for that
in channel state $s \in {\cal S}$ we have that $\lambda \in [0,1]$ one
time shares between two policies ${\bm \pi}_1$ and ${\bm \pi}_2$ to
achieve throughputs for $u \in {\cal U}$ given by
$$
r^s_u  = \lambda {r}^{{\bm \pi}_1,s}_u + (1- \lambda) {r}^{{\bm \pi}_2,s}_u.
$$
Consider $u$ we have 
\begin{eqnarray*}
\lefteqn{ r^{s}_{u}  =   \lambda r^{{\bm \pi}_1,s}_u + (1- \lambda) r^{{\bm \pi}_2,s}_u } \\ 
& =& 
\lambda g^s_u (\phi^{{\bm \pi}_1,s}_{u},{\gamma^{{\bm \pi}_1,s}_u}) + 
(1- \lambda) g^s_u (\phi^{{\bm \pi}_2,s}_{u},{\gamma^{{\bm \pi}_2,s}_u})  \\
& \leq & 
g^s_u (\lambda \phi^{{\bm \pi}_1,s}_{u} + (1-\lambda) \phi^{{\bm \pi}_2,s}_{u},~ 
\lambda \gamma^{{\bm \pi}_1,s}_{u} + (1-\lambda) \phi^{\gamma_2,s}_{u} ) \\
& = &  g^s_u (\phi^{{\bm \pi},s}_{u},  \gamma^{{\bm \pi},s}_{u} ),
\end{eqnarray*}
where
$\phi^{{\bm \pi},s}_{u}=  \lambda  \phi^{{\bm \pi}_1,s}_{u} + (1- \lambda) \phi^{{\bm \pi}_2,s}_{u}$ and
$\gamma^{{\bm \pi},s}_{u} =\lambda  \gamma^{{\bm \pi}_1,s}_{u} + (1- \lambda) \gamma^{{\bm \pi}_2,s}_{u}$.
Clearly ${\bm \phi}^{\bm \pi}, {\bm \gamma}^{\bm \pi}$ as given above 
correspond to a policy $\bm{\pi}$ such that $\bm{\pi} \in \Pi^{H,\delta}$ 
since the set is convex. It also follows that 
$ r^{s}_{u}  \leq r^{{\bm \pi},s}_u $, so 
${c^s_u}  \leq {c}^{{\bm \pi},s}_u$ and so ${\bf c} \leq {\bf c}^{{\bm \pi}}.$


\subsection{Proof of Theorem~\ref{th:optimality_of_stoch_approx}}
\label{pf:optimality_of_stoch_approx}
The proof requires intermediate lemmas, detailed below. For the ease
of exposition, let us define
$U(\bvec{r}):=\sum_{u \in \mathcal{U}} U_u(r_u)$ and
$\nabla U\brac{\bvec{r}}:= \brac{\frac{\partial U_1(x)}{\partial
    x}\Bigr|_{x_1=\substack{r_1}}, \frac{\partial U_2(x)}{\partial
    x}\Bigr|_{x_2=\substack{r_2}}, \ldots, \frac{\partial
    U_1(x)}{\partial
    x}\Bigr|_{\substack{{x_{\abs{\mathcal{U}}}=r_{\abs{\mathcal{U}}}}}}}^T$. First
we have the following important lemma regarding the stochastic
approximation algorithm.
\begin{lemma}
\label{lm:unbiasedness}
$\bvec{R}(t)=\brac{R_1(t), R_2(t), \ldots, R_{\abs{\mathcal{U}}}}^T$
is an unbiased estimator of
$\underset{\bvec{c} \in \mathcal{C}^{H,
    \delta}}{\text{argmax: }} \nabla U\brac{\bvec{\overline{R}(t)}}^T
\bvec{c}$, i.e.,
\begin{equation}
\label{eq:unbiasedness}
\expect{\bvec{R}(t)}= \underset{\bvec{c} \in \mathcal{C}^{H,
    \delta}}{\text{argmax: }} \nabla U\brac{\bvec{\overline{R}(t)}}^T
\bvec{c}.  
\end{equation}
\end{lemma}
\begin{proof}
Based on the definition of $\mathcal{C}^{H, \delta}$ we can
re-write $\underset{\bvec{c} \in \mathcal{C}^{H,
    \delta}}{\text{max: }} \nabla U\brac{\bvec{\overline{R}(t)}}^T
\bvec{c}$ as follows: 
\begin{align}
\underset{\bvec{\phi}, \bvec{\gamma}}{\max}  \quad \sum_{u \in
  \mathcal{U}}& U^{'}_u\brac{\overline{R}_u(t)}\brac{\sum_{s \in \mathcal{S}} p_{\mathcal{S}}
                (s)g^s_u \brac{ \phi^s_u,  \gamma^s_u } }, \\ 
\text{s.t.} \quad 
   \bvec{\phi} & \geq \brac{1-\delta}\bvec{\gamma}, \\
\bvec{\phi},\ \bvec{\gamma} & \in \Pi^{H,\delta}.
\end{align}
Observe that the above optimization problem can be solved separately
for each network state $s \in \mathcal{S}$. The de-coupled problem for
any state $s$ is same as the optimization
problem~\eqref{eqn:opt-stoch-approx} in our online algorithm.  With a
slight abuse of notation, let $\brac{\bvecgreek{\tilde{\phi}}(s),
  \bvecgreek{\tilde{\gamma}}(s)}$ be the optimal solution to the online problem
when $S(t)=s$. Conditioned on $S(t)=s$, we have that: 
\begin{multline}
\expect{R_u(t) \mid S(t)=s}=\expect{f_u^s\brac{\tilde{\phi}_u^s,
    \tilde{\gamma}_u^s D} \mid S(t)=s}\\ =g^s_u \brac{
  \tilde{\phi}^s_u,  \tilde{\gamma}^s_u }  \quad \forall u \in
\mathcal{U}.   
\end{multline}
Computing $\expect{\expect{R_u(t) \mid S(t)}}$ gives the desired
result~\eqref{eq:unbiasedness}.  
\end{proof}

The main intuition behind the proof of optimality is that for large
$t$, the trajectories of $\overline{\bvec{R}}(t)$ can be approximated
by the solution to the following differential equation in
$\bvec{x}(t)$ with continuous time $t$: 
\begin{equation}
\label{eq:differential_equation}
\frac{d \bvec{x}(t)}{dt}= \underset{\bvec{c} \in
  \mathcal{C}^{\mathcal{H}, \delta}}{\text{argmax: }} \nabla
U\brac{\bvec{x}(t)}^T \bvec{c} - \bvec{x}(t).  
\end{equation}
Let us  define $q(\bvec{x}):= \underset{\bvec{c} \in
  \mathcal{C}^{\mathcal{H}, \delta}}{\text{argmax: }} \nabla
U\brac{\bvec{x}}^T \bvec{c}$. To show the optimality of our online
algorithm, we shall also require the following result on the above
differential equation. 
\begin{lemma}
\label{lm:lemma_on_differetial_equation}
The differential equation~\eqref{eq:differential_equation} is globally
asymptotically stable. Furthermore, for any initial condition
$\bvec{x}(0) \in C^{\mathcal{H}, \delta}$, we have that $\lim_{t
  \rightarrow \infty} \bvec{x}(t)= \bvec{r}^*$. 
\end{lemma}
\begin{proof}
  To prove this lemma it is enough to show that there exists a
  Lyapunov function $L(\bvec{x}(t))$ such that it has a negative drift
  when $x(t)\neq \bvec{r}^*$ and has zero drift when
  $x(t)=\bvec{r}^*$.  Define
  $L(\bvec{x})= U(\bvec{r}^*)- U(\bvec{x})$. Observe that under our
  assumption of strictly concave $U_u(\cdot)$, the offline
  optimization problem is guaranteed to have an unique optimal
  solution, which is $\bvec{r}^*$. Therefore,
  $\forall \bvec{x} \in C^{\mathcal{H}, \delta}$ and
  $\bvec{x} \neq \bvec{r}^*$ $L(\bvec{x}) > 0 $. Next we will compute
  the drift of $L(\bvec{x}(t))$ with respect to time.
\begin{align}
\frac{d L(\bvec{x}(t)) }{dt} &= - \nabla U\brac{\bvec{x}(t)}^T \frac{d \bvec{x}(t)}{dt}, \\
							&= -
                                                   q\brac{\bvec{x}(t)}
                                                   + \nabla
                                                   U\brac{\bvec{x}(t)}^T\bvec{x}(t), \label{eq:inequality_of_lyapunov1}\\  
							& < 0 \quad \quad \forall \bvec{x}(t) \neq \bvec{r}^*. 
							\label{eq:inequality_of_lyapunov}
\end{align}
To get inequality~\eqref{eq:inequality_of_lyapunov}, first observe
that from the definition of $q(\bvec{x(t)})$
and~\eqref{eq:inequality_of_lyapunov1}, we get that $\frac{d
  L(\bvec{x}(t)) }{dt}  \leq 0$. However, we have to show that this
inequality is strict for $\bvec{x}(t) \neq \bvec{r}^*$. Observe that
$q(\bvec{x})=\bvec{x}$ is a necessary and sufficient condition for
optimality of the offline optimization problem, see~\cite{boyd_book}
for more details. From strict concavity of the utility functions, we
have an unique optimal point $\bvec{r}^*$. Therefore, $\frac{d
  L(\bvec{x}(t)) }{dt}  <  0$ for  $\bvec{x}(t) \neq \bvec{r}^*$ and
$\frac{d L(\bvec{x}(t)) }{dt}=0$ at $\bvec{x}(t)=\bvec{r}^*$.  
\end{proof}
To conclude the proof, Lemmas~\ref{lm:unbiasedness}
and~\ref{lm:lemma_on_differetial_equation} along with the
condition~\ref{eq:condition_on_epsilon}  satisfy all the conditions
necessary to apply Theorem~2.1 in Chapter~5,~\cite{Kushner_book} which
states that $\bvec{\overline{R}}(t)$ converges to $\bvec{r}^*$ almost
surely.  
\subsection{Proof of Theorem~\ref{thm:optimality_of_minislot_homogeneous_policy}}
\label{pf:optimality_of_minislot_homogeneous_policy}
	The proof has the following  two steps. 
	\begin{enumerate}
		\item  We shall first consider a hypothetical \emph{non-casual} scenario and show that there exists an optimal joint scheduling policy with  minislot-homogeneous  URLLC placement policy which in general is a function of the aggregate URLLC load in an eMBB slot. We then upper bound the optimal value of $\mathcal{OP}_1$ by the solution to a hypothetical {non-causal} scenario described in the sequel. 
		\item Secondly, under Assumption~\ref{asm:assumption_on_convex_cost_functions} on the loss functions, we show that there exists an  URLLC placement policy policy which is  minislot-homogeneous  but  independent of the aggregate URLLC load for the hypothetical non-causal scenario.   We then conclude that there exists an optimal minislot-homogeneous  joint sceduling policy for $\mathcal{OP}_1$ as an upper bound for its value is attained by a minislot-homogeneous  joint scheduling policy. 
	\end{enumerate}
	The two steps are elaborated next. 
	\subsubsection{Hypothetical non-causal scenario}  
	First let us describe the \emph{non-causal} scenario.  At the beginning of each eMBB slot, first  the scheduler chooses $\phi^{\pi, s}$.  Next the total URLLC demand in each minislot is revealed, i.e., the realizations of $D(1), D(2), \ldots, D({\abs{\mathcal{M}}})$ are revealed. Therefore, this setting is not causal as it assumes knowledge about future URLLC demand realizations.  In general the URLLC placement under the non-causal setting is dependent on the minislot index $m$ and $\bvec{D}^{\brac{1:\abs{\mathcal{M}}}}$.  With slight abuse of notation, we shall denote it by $\gamma_{u,m}^s\brac{\bvec{D}^{\brac{1:\abs{\mathcal{M}}}}}$. The joint scheduling policy has to satisfy the constraints~\eqref{eq:constraint_on_phi},~\eqref{eq:constraint_on_gamma_1}, and~\eqref{eq:constraint_on_gamma_2}.   We have the following lemma on the \emph{non-causal} setting. 
	\begin{lemma}
		\label{lm:lemma_on_non_causal}
		There exists an optimal minislot-homogeneous  policy for the  non-casual setting
		such that the URLLC placement  depends only on the total URLLC demand in an eMBB slot, i.e., $\sum_{m=1}^{\abs{\mathcal{M}}}D_m$. 
	\end{lemma}
	\begin{proof}
		Let $\brac{\tilde{\phi}^{\pi}, \tilde{\gamma}^{\pi, s}\brac{ \cdot}}$ be the decision variables under an optimal joint scheduling policy $\pi$ in the non-causal setting. Let $d(1)$, $d(2), \ldots, d({\abs{\mathcal{M}}})$ be realizations of $D(1), D(2), \ldots, D({\abs{\mathcal{M}}})$ such that $\sum_{m=1}^{\abs{\mathcal{M}}} d(m) =d$. Define the following:
		\begin{equation}
		\nu_u^s:=\frac{\sum_{m=1}^{\abs{\mathcal{M}}} \tilde{\gamma}_{u,m}^{\pi, s}\brac{ \bvec{d}^{\brac{1:\abs{\mathcal{M}}}}} d(m)}{d}.
		\end{equation}
		Note that with the definition of $\nu_u^s$, the total puncturing experienced by an eMBB user $u$ in an eMBB slot is $\nu_u^s d$.  From this one can construct an equivalent minislot-homogeneous  URLLC placement policy. For all minislots, use $\nu^s$ as the URLLC placement factor. This satisfies the constraints~\eqref{eq:constraint_on_phi},~\eqref{eq:constraint_on_gamma_1}, and~\eqref{eq:constraint_on_gamma_2}.  In general $\nu^s$ could depend on  $d(1),\, d(2), \, \ldots, d({\abs{\mathcal{M}}})$. However, we will show that the optimal solution  depends only on the sum $\sum_{m=1}^{\abs{\mathcal{M}}} d_m$. 
		
		Let $d'(1), \, d'(2), \ldots,\, d'({\abs{\mathcal{M}}})$ be such that $\sum_{m=1}^{\abs{\mathcal{M}}}d_m'=d$ and there exists an $m$ such that $d'(m)\neq d (m)$. Define the following:
		\begin{equation}
		\nu_u'^s:=\frac{\sum_{m=1}^{\abs{\mathcal{M}}} \tilde{\gamma}_{u,m}^{\pi, s}\brac{\bvec{d}'^{\brac{1:\abs{\mathcal{M}}}}} d_m'}{d}.
		\end{equation}
		Therefore, the total puncturing observed by $\nu_u'^sd$. Observe that $\nu'^s$ is also a feasible URLLC policy for the case when the URLLC demand realizations are $d(1)$, $d(2), \, \ldots, \, d({\abs{\mathcal{M}}})$. Similarly  $\nu^s$ is also a feasible URLLC placement policy for the case with   $d'(1)$, $d'(2), \, \ldots, \, d'({\abs{\mathcal{M}}})$. Therefore, the optimal solution has to be independent of the realizations of $D(1),  \, D(2), \ldots, D({\abs{\mathcal{M}}})$ and depends only on the sum $\sum_{m=1}^{\abs{\mathcal{M}}}D_m$.
	\end{proof}
	
	Therefore, we shall restrict ourselves to minislot-homogeneous  policies in the non-causal setting with the URLLC placement as a function of the total URLLC demand for that eMBB slot. With slight abuse of notation we shall denote a URLLC placement policy in this setting by $\gamma_u^s \brac{\cdot}$ with the only argument as the total URLLC demand in that eMBB slot.  This procedure is formally described next.
	\begin{enumerate}
		\item At the beginning of an eMBB slot, the joint scheduler chooses $\phi_u^{\pi, s}, u \in \mathcal{U}$ such that
		\begin{equation}
		\label{eq:constraint_on_phi_hyp}
		\sum_{u \in \mathcal{U}} \phi^{\pi, s}_u =1  \mbox{  and  } \phi^{\pi, s}_u \in \sbrac{0,1} \quad \forall u. 
		\end{equation} 
		\item The total URLLC demand $D=\sum_{m=1}^{\abs{\mathcal{M}}} D(m)$ in that eMBB slot is revealed.
		\item For an URLLC demand of $D$, $\gamma_u^{\pi, s}(D)$ is chosen such that 
		\begin{equation}
		\sum_{u \in \mathcal{U}}\gamma_u^{\pi, s}(D)=1, \quad \mbox{and} \quad \gamma_u^{\pi, s}(D) \in \sbrac{0,1}. 
		\end{equation} 
	\end{enumerate}     
	Let us denote the feasible policies for this hypothetical non-causal scenario by $\Pi^{\dagger}$. $\brac{\bvecgreek{\phi^{\pi, s}}, \bvecgreek{\gamma^{\pi, s}} }$ is chosen as the solution to the following optimization problem. 
	\begin{equation}
	\mathcal{OP}_2 : \quad
	\underset{\pi \in \Pi^{\dagger}}{\max}: \sum_{u \in \mathcal{U}} w_u g_u^{\pi, s}\brac{\phi_u^{\pi, s}, \gamma_u^{\pi, s}\brac{\cdot}},
	\end{equation}
	where $g_u^{\pi, s}\brac{\phi_u^{\pi, s}, \gamma_u^{\pi, s}\brac{\cdot}} = r_u^s \phi_u^{\pi, s} \expect{{1-h_u^s\brac{\frac{\gamma_u^{\pi, s}(D)D}{\phi_u^{\pi, s}}} }}$. 
	We have the following important lemma which states that the optimal value under the non-causal scenario is an upper bound to the optimal value under the causal and minislot-dependent  policy.
	\begin{lemma}
		\label{lm:lemma_on_upper_bound}
		\begin{multline}
		\underset{\pi \in \Pi^{\dagger}}{\max}: \sum_{u \in \mathcal{U}} w_u g_u^{\pi, s}\brac{\phi_u^{\pi, s}, \gamma_u^{\pi, s}\brac{\cdot}} \\ \geq \underset{\pi \in \tilde{\Pi}}{\max}: \sum_{u \in \mathcal{U}} w_u g_u^{\pi, s}\brac{\phi_u^{\pi, s}, \bvecgreek{\gamma}_{u}^{\pi, s}}.
		\end{multline}
	\end{lemma}
	\begin{proof}
		This directly follows from the proof of Lemma~\ref{lm:lemma_on_non_causal} where we have shown that any URLLC placement factor $\bvecgreek{\gamma}_{u}^{\pi, s}$ can be transformed into a minislot-homogeneous  policy which depend only on the total URLLC demand in an eMBB slot, and hence, any  feasible solution for $\mathcal{OP}_1$ is a feasible solution for $\mathcal{OP}_2$.
	\end{proof}
	
	\subsubsection{Existence of an optimal solution independent of the value of $ D $}  	
	In general the optimal URLLC placement  policy under $\mathcal{OP}_2$ may depend on the total URLLC demand in an eMBB slot. However, under the Assumption~\ref{asm:assumption_on_convex_cost_functions} it is independent of the total URLLC demand.  This is stated formally in the following lemma. 
	\begin{lemma}
		\label{lm:optimality_on_lemma_on_upper_bound}
		Under Assumption~\ref{asm:assumption_on_convex_cost_functions}, there exists an optimal  solution  $\brac{\bvec{\phi}^{*, s}, \bvec{\gamma}^{*, s}\brac{\cdot}}$ for $\mathcal{OP}_2$ with URLLC placement policy ($\bvec{\gamma}^{*, s}\brac{\cdot}$)  independent of $D$.
	\end{lemma} 
	\begin{proof}
		If $\brac{\bvec{\phi}^{*, s}, \bvec{\gamma}^{*, s}\brac{\cdot}}$ is an optimal solution to $\mathcal{OP}_2$, then $\bvec{\gamma}^{*, s}\brac{\cdot}$ must also be an optimal solution to the following optimization problem in $\bvecgreek{\gamma}^s:=\brac{\gamma_1^s(\cdot), \gamma_2^s(\cdot), \ldots, \gamma_{\abs{\mathcal{U}}}^s(\cdot) }$. 
		\begin{align}
		\underset{\bvecgreek{\gamma}^s}{\max} \quad \sum_{u \in \mathcal{U}}&
		w_u
		g^s_u
		(
		\phi^{*, s}_u,
		\gamma^s_u\brac{\cdot}
		)
		, \\ 
		\text{s.t.} \quad 
		\phi_u^{*, s} & \geq \brac{1-\delta}\gamma_u^{s}(d) \quad  \forall u,\, d, \\
		\sum_{u \in \mathcal{U}} \gamma^{s}_u(d) &=1 \mbox{  and  } \gamma_u^s(d) \in \sbrac{0,1} \quad \forall u,\, d. \\
		\end{align}
		For any $d$ and $u$, from the K.K.T. conditions for the above optimization problem, we have that
		\begin{equation}
		\label{eq:KKT_1}
		-w_u r_u^s d^p h_u^{s'}\brac{\frac{\gamma_u^{*, s}(d) }{\phi_u^{*, s}}} + \beta(d) + \eta_u(d)  -\nu_u(d) -\lambda_u(d)=0.
		\end{equation}
		where $h_u^{s'}(x) = \frac{d h_u^{s}(y)}{dy}\Bigr|_{\substack{y=x}},   $$\beta(d)$ is an arbitrary constant (function of $d$) and $\eta_u(d)$, $\nu_u(d)$ and $\lambda_u(d)$ are constants such that 
		\begin{align}
		\lambda_u(d) \brac{{\phi}_u^{*, s}(d) - {\gamma}_u^{*, s} (1-\delta)}= 0 \quad & \text{and }  \quad \lambda_u(d) \geq 0 \quad \forall u, \\
		\eta_u(d)  {\gamma}_u^{*, s} (d) =0 \quad & \text{and}  \quad \eta_u(d) \geq  0 \quad \forall u, \\
		\nu_u(d) \brac{1-{\gamma}_u^{*, s}(d)} =0 \quad & \text{and}  \quad \nu_u(d) \geq  0 \quad \forall u.
		\end{align}
		Note that we have used the fact that for a homogeneous loss functions $ h_u^{s'} (dx)=d^p h_u^{s'}(x) $. 
		For any $\tilde{d} \neq d$, if we choose $\beta(\tilde{d})=\beta(d)\frac{\tilde{d}^p}{d^p}$, $\eta_u(\tilde{d})= \eta_u(d)\frac{\tilde{d}^p}{d^p} $, $\nu_u(\tilde{d})=\nu_u(d) \frac{\tilde{d}^p}{d^p} $, and $\lambda_u(\tilde{d})= \lambda_u(d) \frac{\tilde{d}^p}{d^p}$, then from~\eqref{eq:KKT_1} $\gamma_u^{*, s} (d)$ and $\phi_u^{*, s}$ satisfy the K.K.T. condition for $ \tilde{d} $
		\begin{equation}
		-w_u r_u^s  \tilde{d}^p h^{s'}\brac{\frac{\gamma_u^{*, s}(d)}{\phi_u^{*, s}}} + \beta(\tilde{d}) + \eta_u(\tilde{d})  -\nu_u(\tilde{d}) -\lambda_u(\tilde{d})=0.
		\end{equation} 
		Hence, $\gamma_u^{*, s} (d)$ and $\phi_u^s$ are optimal for $\tilde{d}$ too.  Hence, we have a constructed an optimal solution with URLLC placement policy independent of $D$.   
	\end{proof}
	
	We have shown in Lemma~\ref{lm:optimality_on_lemma_on_upper_bound} that there exists an optimal policy $\brac{\bvec{\phi^{*, s}}, \bvec{\gamma^{*, s}}}$ which is a minislot-homogeneous  policy and independent of the realization of $D$. In Lemma~\ref{lm:lemma_on_upper_bound}, we have also shown that the optimal value of $\mathcal{OP}_2$ is an upper bound for $\mathcal{OP}_1$. Hence, there exists a minislot-homogeneous  policy which achieves an upper bound for $\mathcal{OP}_1$. Therefore, there exists a minislot-homogeneous  policy which is optimal for $\mathcal{OP}_1$.  

\subsection{Proof of Theorem~\ref{thm:horizontal_vs_vertical}}
\label{pf:horizontal_vs_vertical}
	Let $\mathcal{S}_k$ be the set of all subsets with $k$ elements chosen from the set $\cbrac{1,2, \ldots, m_1+m_2}$. For example, if $m_1+m_2=3$ and $k=2$, then $\mathcal{S}_k =\cbrac{\cbrac{1,2}, \cbrac{2,3}, \cbrac{1,3}}$. Note that  $\abs{\mathcal{S}_k}= \binom{\abs{\mathcal{M}}}{k}$.  Using the above definitions, we can re-write the R.H.S. of~\eqref{eq:result_on_vertical_vs_horizontal} as follows:
	\begin{multline}
	\expect{h_1^s \brac{\sum_{m=1}^{m_1 + m_2 } \phi_1 D(m)}} \\ = \expect{h_1^s\brac{ \frac{1}{ \binom{m_1 + m_2}{m_1}} \sum_{q \in \mathcal{S}_{m_1}}\brac{\sum_{m \in  q} D(m)}}}.
	\end{multline}
	Using the above expression one can apply Jensen's inequality on the R.H.S. of~\eqref{eq:result_on_vertical_vs_horizontal}, we have that
	\begin{multline}
	\expect{h_1^s \brac{\sum_{m=1}^{m_1 + m_2} \phi_1 D(m)}} \\ \leq \frac{1}{ \binom{m_1 + m_2}{m_1}} \sum_{q \in \mathcal{S}_{m_1}} \expect{h_1^s\brac{  {\sum_{m \in  q} D(m)}}}. 
	\end{multline}
	Since $D_m$'s are i.i.d. the R.H.S. of the above expression is same as the L.H.S. of~\eqref{eq:result_on_vertical_vs_horizontal}. Hence, proved.

\subsection{Proof of Theorem~\ref{thm:tp-optimality-theorem}} 
\label{pf:tp-optimality-theorem}

  Clearly the probability of loss depends on the minislot demands and
  the users thresholds. If one relaxes the sequential constraint on
  URLLC allocations, one can consider aggregating the the minislot
  demands and pooling together the users superposition/puncturing
  thresholds. The probability of loss for this relaxed system is
  simply the probability the demand exceeds the size of the
  superposition/puncturing pool, i.e., The probability of loss under
  the pooled resources is given by
$$
P(D \geq \sum_{u \in {\cal U}} \phi_{u}^{s} t_{u}^{s}(\phi_{u}^{s})).
$$
This is clearly a lower bound for any placement policy. Note however
that the threshold proportional strategy meets this bound from
Corollary~\ref{cor:tp} (see
Equation~\eqref{eqn:tp-bound}) so it indeed minimizes the probability of
loss on a given eMBB slot.

\end{document}